\documentclass[11pt]{article}
\usepackage{amsmath,amssymb,amsfonts,amsthm,epsfig}
\usepackage[usenames,dvipsnames]{xcolor}
\usepackage{bm,xspace}
\usepackage{tcolorbox}
\usepackage{cancel}
\usepackage{fullpage}
\usepackage{liyang}
\usepackage{framed}
\usepackage{verbatim}
\usepackage{enumitem}
\usepackage{array}
\usepackage{multirow}
\usepackage{afterpage}
\usepackage{mathrsfs}
\usepackage{pifont} 
\usepackage{chngpage}
\usepackage[normalem]{ulem}
\usepackage{boxedminipage}
\usepackage{caption}
\usepackage{subcaption}
\usepackage{bold-extra}
\usepackage{forest}

\usepackage{algorithm}
\usepackage{algorithmicx}
\usepackage{algpseudocode}

\usepackage{tikz}
\usetikzlibrary{calc,through,backgrounds,decorations.pathreplacing, calligraphy,arrows.meta}
\usetikzlibrary{positioning,chains,fit,shapes}

\usepackage{tikz-cd}

\usepackage{pgfplots}
\pgfplotsset{width=8cm,compat=newest}

\def\colorful{0}

\ifnum\colorful=1
\newcommand{\violet}[1]{{\color{violet}{#1}}}

\fi
\ifnum\colorful=0
\newcommand{\violet}[1]{{{#1}}}

\fi

\newtheorem{hypothesis}{Hypothesis}

\newlist{enumprop}{enumerate}{1} % set up a dedicated enumeration environment
\setlist[enumprop]{label=\arabic*.,ref=\theproposition.\arabic*}
\crefalias{enumpropi}{proposition}

\makeatletter
\newtheorem*{rep@theorem}{\rep@title}
\newcommand{\newreptheorem}[2]{
\newenvironment{rep#1}[1]{
 \def\rep@title{#2 \ref{##1}}
 \begin{rep@theorem}\itshape}
 {\end{rep@theorem}}}
\makeatother

\newreptheorem{theorem}{Theorem}

\newcommand{\BlockwisePar}{\textnormal{BlockwisePar}}

\begin{document}

%!TEX root = paper.tex

\title{Fast decision tree learning solves hard coding-theoretic problems\vspace{10pt}}

\author{ 
Caleb Koch \vspace{6pt} \\ 
{{\sl Stanford}} \and 
\hspace{5pt} Carmen Strassle \vspace{6pt} \\
\hspace{5pt} {{\sl Stanford}} \vspace{10pt} \and 
Li-Yang Tan \vspace{6pt}  \\
\hspace{-10pt} {{\sl Stanford}}
}

\date{\small{\today}}

 \maketitle

 \begin{abstract}
We connect the problem of properly PAC learning decision trees to the parameterized {\sc Nearest Codeword Problem} ($k$-NCP). Despite significant effort by the respective communities, algorithmic progress on both problems has been stuck: the fastest known algorithm for the former runs in quasipolynomial time (Ehrenfeucht and Haussler 1989) and the best known approximation ratio for the latter is $O(n/\log n)$ (Berman and Karpinsky 2002; Alon, Panigrahy, and Yekhanin 2009). Research on both problems has thus far proceeded independently with no known connections. 

We show that {\sl any} improvement of Ehrenfeucht and Haussler’s algorithm will yield $O(\log n)$-approximation algorithms for $k$-NCP, an exponential improvement of the current state of the art. This can be interpreted either as a new avenue for designing algorithms for $k$-NCP, or as one for establishing the optimality of Ehrenfeucht and Haussler’s algorithm. Furthermore, our reduction along with existing inapproximability results for $k$-NCP already rule out polynomial-time algorithms for properly learning decision trees. A notable aspect of our hardness results is that they hold even in the setting of {\sl weak} learning whereas prior ones were limited to the setting of strong learning. 
\end{abstract} 

 \thispagestyle{empty}
 \newpage

 \thispagestyle{empty}
 \setcounter{tocdepth}{2}
 \tableofcontents
 \thispagestyle{empty}
 \newpage 

 \setcounter{page}{1}

\section{Introduction}

This paper connects two fundamental problems from two different areas, learning theory and coding theory.

\medskip

\begin{tcolorbox}[colback = white,arc=1mm, boxrule=0.25mm]
\paragraph{Properly PAC Learning Decision Trees ({\sc DT-Learn}).} Given random examples generated according to a distribution $\mathcal{D}$ and labeled by a function $f$, find a small decision tree that well-approximates~$f$. 
\end{tcolorbox}
\medskip 

The fastest known algorithm for this problem is due to Ehrenfeucht and Haussler from 1989 and runs in quasipolynomial time:

\newtheorem*{theorem*}{Theorem}

\begin{theorem*}[\cite{EH89}]
 There is an algorithm that, given random examples $(\bx,f(\bx))$ where $f : \zo^n \to \zo$ is a size-$s$ decision tree and $\bx$ is drawn from a distribution $\mathcal{D}$ over $\zo^n$, runs in $\poly(n^{\log s},1/\eps)$ time and returns a decision tree $T$ that is $\eps$-close to $f$ under~$\mathcal{D}$. \end{theorem*}
 
There are no known improvements to~\cite{EH89}'s algorithm even in the setting of {\sl weak} learning where $T$ only has to be mildly correlated with $f$ (i.e.~for values of $\eps$ close to $\frac1{2}$).  

\medskip

\begin{tcolorbox}[colback = white,arc=1mm, boxrule=0.25mm]
\paragraph{Parameterized Nearest Codeword Problem ({\sc $k$-NCP}).} Given the generator matrix of a linear code of dimension $n$, a received word $z$, and a parameter $k$, decide if there is a codeword within Hamming distance $k$ of $z$. 
\end{tcolorbox}
\medskip 

This problem is $\mathrm{W}[1]$-hard~\cite{DFVW99},  so it is natural to seek approximation algorithms. The current best algorithms achieve an approximation ratio of $O(n/\log n)$:

\begin{theorem*}[\cite{BK02,APY09}]
\label{thm:approx algs for k-NCP}
There is an algorithm that, given the generator matrix of a linear code $\mathcal{C}$ of dimension~$n$, a received word $z$, a parameter $k$, and the promise that there is a codeword of $\mathcal{C}$ within distance $k$ of $z$, runs in polynomial time and returns a codeword within distance $\alpha k$ of~$z$ where $\alpha = O(n/\log n)$.  
\end{theorem*}

Berman and Karpinsky's algorithm is randomized whereas Alon, Panigrahy, and Yekhanin's is deterministic. Note that $k$-NCP can be solved exactly (i.e.~with $\alpha = 1$) in time $n^{O(k)}$. There are no known algorithms that run in time $n^{o(k)}$ and achieve an approximation ratio of $\alpha = o(n/\log n)$.

\subsection{Motivation for both problems}

Both problems are central and well-studied in their respective fields of learning theory and coding theory. Part of the theoretical interest in {\sc DT-Learn}---specifically, {\sl proper} learning of decision trees---stems from the role that decision trees play in machine learning practice. They are the prime example of an interpretable hypothesis, and a recent survey of interpretable machine learning~\cite{RCCHSZ22} lists the problem of constructing optimal decision tree representations of data as the very first of the field's “10 grand challenges”. 

\cite{EH89}’s algorithm was one of the earliest PAC learning algorithms, coming on the heels of Valiant’s introduction of the model~\cite{Val84}. Numerous works have since obtained faster algorithms for variants of the problem~\cite{Bsh93,KM93,SS93,JS05,OS07,GKK08,KST09,BLT-ITCS,BLQT22,Boi24}, but~\cite{EH89}’s algorithm for the original problem has  resisted improvement.  Indeed, faster algorithms for {\sc DT-Learn} are known to have significant consequences within learning theory. Even just under the uniform distribution, {\sc DT-Learn} contains as a special case the {\sl junta problem}~\cite{Blu94,BL97}, which itself has been called “the most important problem in uniform distribution learning”~\cite{MOS04}. Since every $k$-junta is a decision tree of size $s \le 2^k$, an $n^{o(\log s)}$ time algorithm for {\sc DT-Learn} gives an $n^{o(k)}$ time algorithm for learning $k$-juntas---this would be a breakthrough, as the current fastest algorithms run in $n^{ck}$ time for some constant $c < 1$~\cite{MOS04,Val15}. Far less is known about connections between {\sc DT-Learn} and problems {\sl outside} of learning theory. 

The {\sc Nearest Codeword Problem} (NCP), also known as {\sc Maximum Likelihood Decoding}, has been called “probably the most fundamental computational problem on linear codes”~\cite{Mic}. While   specific codes are often designed in tandem with fast decoding algorithms, results on the general problem have skewed heavily towards the side of hardness. NCP was proved to be NP-complete by Berlekamp, McEliece, and van Tilborg in 1978~\cite{BMvT78}. Various aspects of its complexity has since been further studied in multiple lines of work, including the hardness of approximation~\cite{ABSS97,DKS98,DKRS03,DMS03,Ale11}; hardness with preprocessing~\cite{BN90,Lob90,Reg03,FM04,GV05,AKKV11,KPV14}; hardness under ETH and SETH~\cite{BIWX11,SDV19}; and most relevant to this work, hardness in the parameterized setting~\cite{DFVW99,ALW14,BELM16,BGKM18,Man20,LRSW22,BCGR23,LLL24,GLRSW24}. On the other hand, the only known algorithms are those of~\cite{BK02,APY09}.

\section{Our results}

We show how algorithms for {\sc DT-Learn} yield approximation algorithms for {\sc $k$-NCP}. Before stating our result in its full generality (\Cref{thm:most general} below), we first list a couple of its consequences. One instantiation of parameters shows that {\sl any} improvement of~\cite{EH89}'s runtime, even in the setting of weak learning, will give new approximation algorithms for $k$-NCP with exponentially-improved approximation ratios: 

\begin{corollary}
\label{cor:improving EH gives new NCP algorithms} 
Suppose there is an algorithm that given  random examples generated according to a distribution $\mathcal{D}$ over $\zo^n$ and labeled by a size-$s$ decision tree runs in time $n^{o(\log s)}$ and w.h.p.~outputs a decision tree with accuracy $\frac{1}{2}+\frac{1}{\poly(n)}$ under $\mathcal{D}$. Then for $k =  \Theta(\log s)$ there is a randomized algorithm running in time $n^{o(k)}$ which solves $O(\log n)$-approximate $k$-\textnormal{NCP}. 
\end{corollary}

A different instantiation of parameters shows that a {\sl polynomial-time} algorithm for properly learning decision trees, again even in the setting of weak learning, will give {\sl constant-factor} approximation algorithms for $k$-NCP. Since the latter has been ruled out under standard complexity-theoretic assumptions~\cite{BELM16,Man20,LLL24}, we get: 

\begin{corollary}
\label{cor:W1 hardness of polytime}
    % Assume randomized ETH. Then [...]
    Assuming $\mathrm{W}[1] \ne \mathrm{FPT}$, there is no polynomial-time algorithm for properly learning decision trees, even in the setting of weak learning. 
\end{corollary}

That is, there is no algorithm that given random examples generated according to a distribution~$\mathcal{D}$ over $\zo^n$ and labeled by a size-$n$ decision tree, runs in $\poly(n)$ time and w.h.p.~outputs a decision tree hypothesis that achieves accuracy $\frac1{2} + \frac1{\mathrm{poly}(n)}$ under~$\mathcal{D}$. Prior to our work, there were no results ruling out polynomial-time algorithms achieving error $\eps =  0.01$, much less $\eps = \frac1{2} - o(1)$. \violet{See \Cref{fig:results} for an illustration of our results.}

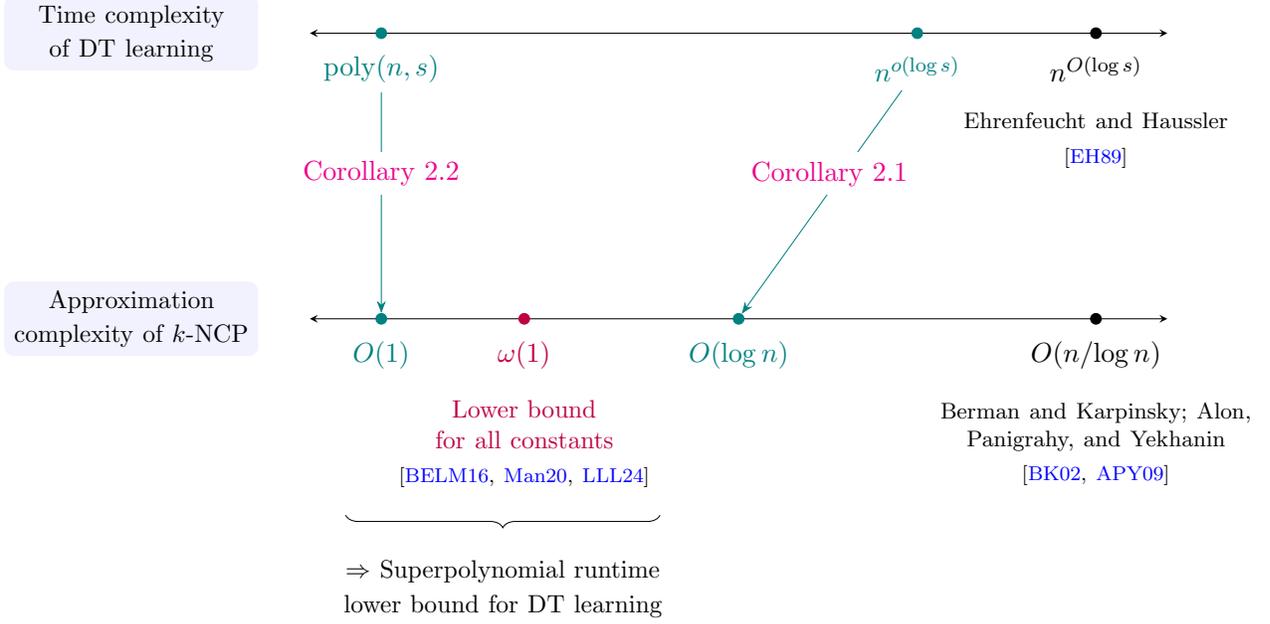
\begin{figure}[t]
    \centering
    \scalebox{0.95}{
    \begin{tikzpicture}[]
        % x DEFINES HOW WIDE FIGURE IS
        \def\x{6}
        % xh defines third width
        \def\xh{1.83}
        % xt defines two-thirds width
        \def\xt{3.67}
        % xf defines four-thirds width
        \def\xf{7.5}
        % c DEFINES CENTER
        \def\c{0}
        % l DEFINES LEFT CENTER
        \def\l{-0.1}
        % l DEFINES RIGHT CENTER
        \def\r{0.1}

        % MAIN TWO HORIZONTAL AXES
        \draw[black,stealth-stealth] (-\x,0) node[left] {{}} -- (\x,0) node[left]{{}};
        \draw[black,stealth-stealth] (-\x,4) node[left] {{}} -- (\x,4) node[left]{{}};
     
        % RIGHT AXES POINTS
        % \node[draw,circle,fill=black,inner sep=1.5pt,xshift=\x cm] (k'R) at ($(\c,\yt) + (0,3*\inc)$) {};
        % \node[draw,circle,fill=black,inner sep=1.5pt,xshift=\x cm] (kR) at (\c,-1) {};
        
        % LABELS FOR AXES
        \draw[xshift=-1 cm] (-\xf,0)  node [rounded corners, text centered, fill=blue!5, text width=3.3 cm] {\small {Approximation complexity of $k$-NCP} };
        
        \draw[xshift=-1 cm] (-\xf,4)  node [rounded corners, text centered, fill=blue!5, text width=3.3 cm] {\small {Time complexity of  DT learning} };

        % DT labels
        \node[draw,circle,fill=black,inner sep=1.5pt] (EH) at (5,4) {};
        \node[draw,circle,fill=black,inner sep=1.5pt,color=teal] (GETH) at (2.5,4) {};
        \node[draw,circle,fill=black,inner sep=1.5pt,color=teal] (POLY) at (-5,4) {};

        % NCP labels
        \node[draw,circle,fill=black,inner sep=1.5pt] (N) at (5,0) {};
        \node[draw,circle,fill=black,inner sep=1.5pt,color=teal] (LOGN) at (0,0) {};
        \node[draw,circle,fill=black,inner sep=1.5pt,color=teal] (CON1) at (-5,0) {};
        \node[draw,circle,fill=black,inner sep=1.5pt,color=purple] (CON2) at (-3,0) {};

        % Lines
        \draw[teal,-Stealth] {[yshift=-0.5 cm](2.5,4)} to node[midway,fill=white] {\Cref{cor:improving EH gives new NCP algorithms}} (LOGN);
        \draw[teal,-Stealth] {[yshift=-0.5 cm](-5,4)} to node[midway,fill=white] {\Cref{cor:W1 hardness of polytime}} (CON1);
        
        % Runtime labels
        \draw[color=black] (EH) node[yshift=-0.5 cm,fill=white] {$n^{O(\log s)}$};
        \draw[color=teal] (GETH) node[yshift=-0.5 cm,fill=white] {$n^{o(\log s)}$};
        \draw[color=teal] (POLY) node[yshift=-0.5 cm,fill=white] {$\poly(n,s)$};

        % APX labels
        \draw[color=black] (N) node[yshift=-0.5 cm,fill=white] {$O(n/{\log n})$};
        \draw[color=teal] (LOGN) node[yshift=-0.5 cm,fill=white] {$O({\log n})$};
        \draw[color=purple] (CON2) node[yshift=-0.5 cm,fill=white] {{$\omega(1)$}};
        \draw[color=teal] (CON1) node[yshift=-0.5 cm,fill=white] {$O(1)$};

        % Citations
        \draw[color=black] (EH) node[yshift=-1.5 cm,fill=white,text width=4 cm, text centered] {\footnotesize Ehrenfeucht and Haussler \\ \scriptsize{\cite{EH89}}};
        \draw[color=black] (N) node[yshift=-1.75 cm,fill=white,text width=4.5 cm, text centered] {\footnotesize Berman and Karpinsky; Alon, Panigrahy, and Yekhanin \\ {\scriptsize\cite{BK02,APY09}}};
        \draw[color=black] (CON2) node[yshift=-1.75 cm,fill=white,text width=4.5 cm, text centered] {\small {\color{purple}{Lower bound for all constants} \\ \color{black}{\scriptsize{\cite{BELM16,Man20,LLL24}}}}};

        \path[draw,decorate,decoration={brace,mirror,amplitude=5pt}] (-5.5,-2.75) -- (-1.1,-2.75) node[midway,text width=5.5 cm, text centered,below=0.5cm]{\small {$\Rightarrow$ Superpolynomial runtime lower bound for DT learning}};
        
    \end{tikzpicture}
    }
      \captionsetup{width=.9\linewidth}
    \caption{An illustration of the implications of our main result. The top axis denotes different runtimes for (weak) learning $n$-variable size-$s$ decision trees. The bottom axis denotes approximation factors for $k$-NCP. The right hand side of each axis plots the best known algorithms for each respective problem. Each arrow indicates how a decision tree learning algorithm with a particular runtime yields an algorithm for $k$-NCP with a corresponding approximation ratio. %Combining known lower bounds for approximating $k$-NCP \cite{BELM16,Man20,LLL24}, we derive new lower bounds for learning decision trees.
    }
    \label{fig:results}
\end{figure}

%\calnote{Do we want to state the hardness we get with Gap-ETH using Manurangsi's NCP hardness? One aspect of this is that it gives the first quasipolynomial lower bound ($n^{\Omega(\log s)}$) under a standard complexity assumption for any weakly proper learning algorithm (i.e. learning size-$s$ with size-$s'$ for $s'>s$). The previous best bound in this setting was from SODA23 where we showed under ETH that learning size-$s$ by size-$s$ takes time $n^{\Omega(\log s)}$.}

\subsection{Statement of our reduction}

\Cref{cor:improving EH gives new NCP algorithms,cor:W1 hardness of polytime} place no restrictions on the size $s'$ of the decision tree hypothesis that the algorithm is allowed to output, other than the obvious one of $s' \le t$ where $t$ is the algorithm's runtime.  The most general statement of our reduction decouples these two quantities. Algorithms that achieve small $s'$ (ideally, close to the size $s$ of the target decision tree) are of interest even if $t$ is not comparably small.  We show: \medskip

\begin{tcolorbox}[colback = white,arc=1mm, boxrule=0.25mm]
\begin{theorem}[Our reduction]
\label{thm:most general}
Suppose there is an algorithm that given random examples generated according to a distribution $\mathcal{D}$ over $\zo^n$ and labeled by a size-$s$ decision tree, runs in time $t(n,s,s',\eps)$ and w.h.p.~outputs a size-$s'$ decision tree hypothesis that achieves accuracy $1-\eps$ under $\mathcal{D}$. Then, for all $\ell\in \N$ there is a randomized algorithm  which solves $\alpha$-approximate $k$-\textnormal{NCP} running in time 
\[ O(\ell n^2)\cdot t(\ell n,2^{\ell k},2^{O(\alpha \ell k)},\eps)+\poly(n,\ell,2^{\alpha \ell k}) \ \text{  
where\ $\eps=\lfrac{1}{2}-2^{-\Omega(\alpha \ell k)}$}.\]  
\end{theorem}
\end{tcolorbox}
\medskip 

(The parameter $\ell$ will be used to pad instances of $k$-NCP for small $k$ to get instances of learning size-$s$ decision trees for large $s$.) 

Decoupling $s'$ and $t$ allows us to show variants of~\Cref{cor:W1 hardness of polytime} where we obtain stronger time lower bounds at the price of stronger complexity-theoretic assumptions:

\begin{corollary}
    \label{cor:explicit superpoly}
    Suppose there is an algorithm which given random examples generated according to a distribution $\mathcal{D}$ and labeled by a size-$n$ decision tree w.h.p.~outputs a decision tree hypothesis of $\poly(n)$ size that achieves accuracy $\frac1{2} + \frac1{\mathrm{poly}(n)}$ under $\mathcal{D}$. Then: 
    \begin{enumerate}
    \item (\Cref{cor:W1 hardness of polytime} restated) If the algorithm runs in $\poly(n)$ time, then $\mathrm{W}[1]=  \mathrm{FPT}$. 
            \item If the algorithm runs in time $n^{(\log n)^\delta}$ for a sufficiently small constant $\delta$, then  \textnormal{ETH} is false.
        \item If the algorithm runs in time $n^{o(\log n)}$, then  \textnormal{Gap-ETH} is false.
    \end{enumerate}
    % Assuming randomized Gap-\textnormal{ETH}, there is no algorithm running in time $n^{o(\log n)}$ which given random examples distributed according to a distribution $\mathcal{D}$ and labeled by a size-$n$ decision tree, outputs a decision tree hypothesis of size $\poly(n)$ that achieves accuracy $\frac1{2} + \frac1{\mathrm{poly}(n)}$ under $\mathcal{D}$. Assuming randomized \textnormal{ETH}, for all constant $c>1$, there is no algorithm which given random examples distributed according to a distribution $\mathcal{D}$ and labeled by a size-$n$ decision tree, outputs a decision tree hypothesis of size $n^c$ that achieves accuracy $\frac1{2} + \frac1{\mathrm{poly}(n)}$ under $\mathcal{D}$ and runs in time $n^{o((\log s)^\delta)}$ where $\delta=\frac{1}{\polylog c}$.
\end{corollary}

\paragraph{Addressing the main open problem from~\cite{EH89}.} Paraphrasing the very first open problem of~\cite{EH89}, the authors ask: 

\begin{quote} 
{\sl For the concept class of polynomial-size decision trees $($i.e.~$s = \mathrm{poly}(n))$, can one design algorithms that run in polynomial time $($i.e.~achieve $t \le \mathrm{poly}(n))$? Failing that, can one at least design algorithms that take superpolynomial time as those given here, but return polynomial-size decision tree hypotheses $($i.e.~achieve $s' \le \mathrm{poly}(n))$?”} 
\end{quote}

\Cref{cor:W1 hardness of polytime} provides a negative answer to the first question and~\Cref{cor:explicit superpoly} provides negative answers to the second, with both holding even in the setting of weak learning.

\subsection{Comparison with prior work}

While we view the connection between {\sc DT-Learn} and $k$-{\sc NCP} as our main contribution, the new lower bounds that we obtain (i.e.~\Cref{cor:W1 hardness of polytime,cor:explicit superpoly}) also compare favorably with existing ones.

\paragraph{Inverse-polynomial error.} There is a long line of work on the hardness of {\sc DT-Learn}   in the regime of inverse-polynomial error, $\eps=1/\poly(n)$. Pitt and Valiant~\cite{PV88} first showed, via a simple reduction from {\sc Set Cover}, that properly learning decision size-$s$ decision trees (where $s=n$) to such an accuracy is NP-hard---{\sl if} the algorithm is additionally required to output a hypothesis whose size $s'$ exactly matches that of the target (i.e.~$s' = s$). Hancock, Jiang, Li, and Tromp~\cite{HJLT96} subsequently ruled out polynomial-time algorithms that are required to return a hypothesis of size $s' \le s^{1+o(1)}$, under the assumption that SAT cannot be solved in randomized quasipolynomial time. Alekhnovich, Braverman, Feldman, Klivans, and Pitassi~\cite{ABFKP09} then ruled out polynomial-time algorithms, now with no restrictions on hypothesis size, under the randomized Exponential Time Hypothesis (ETH). Koch, Strassle, and Tan~\cite{KST23soda} improved~\cite{ABFKP09}'s lower bound to $n^{\Omega(\log \log n)}$ under the randomized ETH, and to $n^{\Omega(\log n)}$ under a plausible conjecture on the complexity of {\sc Set Cover}.   

\paragraph{Constant error.} The above line of work is built successively on~\cite{PV88}'s reduction from {\sc Set Cover}, which appears limited to the setting where $\eps = 1/\poly(n)$.  Recent work of Koch, Strassle, and Tan~\cite{KST23} showed, via a new reduction from {\sc Vertex Cover}, that the problem is NP-hard even for $\eps$ being a small absolute constant $(\eps = 0.01)$. However, their result again only holds if the algorithm is required to output a hypothesis of size $s' = s$, like in the original result of~\cite{PV88}. (The focus of~\cite{KST23}'s work was in giving the first lower bounds against {\sl query} learners, whereas none of the prior work, or ours, applies to such learners.)

\backrefsetup{disable}
\begin{table}[t]
\begin{adjustwidth}{-3.6em}{}
\renewcommand{\arraystretch}{2}
\centering
\begin{tabular}{|c|c|c|c|c|}
\hline
~~~Reference~~~ & 
  \begin{tabular}{@{}c@{}}
~~Restriction on~~ \vspace{-12pt} \\
~~hypothesis size $s'$~~ \end{tabular} 
 &~~~Error $\eps$~~~& ~~ Runtime lower bound~~ & Hardness assumption \\
\hline 
\hline 
\cite{PV88} & $s' = s$ & $ 1/\poly(n)$ & $n^{\omega(1)}$ & $\mathrm{SAT}\notin \mathsf{RP}$  \\ \hline 
\cite{HJLT96} & $s'\le s^{1+o(1)}$ & $ 1/\poly(n)$ & $n^{\omega(1)}$ & ~~$\mathrm{SAT}\notin \mathsf{RTIME}(n^{\polylog(n)})$~~  \\ \hline 
\cite{ABFKP09} & None &  $ 1/\poly(n)$ & $n^{\omega(1)}$ & ETH \\ \hline   \cite{KST23soda} & None &  $ 1/\poly(n)$ & $n^{\Omega(\log\log n)}$ & ETH \\ \hline 
\cite{KST23} & $s' = s$ &  $0.01$ & $n^{\omega(1)}$ & $\mathrm{SAT}\notin \mathsf{RP}$  \\ \hline   \hline 
\Cref{cor:W1 hardness of polytime} & None & $\frac1{2} - \frac1{\poly(n)}$ & $n^{\omega(1)}$ & $\mathrm{W}[1]\ne \mathrm{FPT}$ \\ \hline   
\Cref{cor:explicit superpoly} & $s' \le \poly(s)$ & ~~$\frac1{2} - \frac1{\poly(n)}$~~ & $n^{(\log n)^{\Omega(1)}}$ &  ETH \\ \hline 
\Cref{cor:explicit superpoly} & $s' \le \poly(s)$ & ~~$\frac1{2} - \frac1{\poly(n)}$~~ & $n^{\Omega(\log n)}$ &  Gap-ETH \\ \hline 
\end{tabular} 
\end{adjustwidth}
  \captionsetup{width=.9\linewidth}
\caption{Lower bounds for properly learning $n$-variable size-$s$ decision trees under standard complexity-theoretic assumptions. All of them hold for $s=n$.}
\label{table}
\end{table} 
\backrefsetup{enable}

\paragraph{Summary.}  Prior lower bounds either held for $\varepsilon = 1/\mathrm{poly}(n)$, or for $\varepsilon = 0.01$ under the restriction that $s' = s$. For constant $\eps$ there were no lower bounds for general polynomial-time algorithms (i.e.~ones without any restriction on their hypothesis size), and for $\varepsilon = \frac1{2} - o(1)$, there were no lower bounds even under the strictest possible restriction that $s' = s$. See~\Cref{table}.

As we will soon discuss, the linear-algebraic nature of $k$-NCP is crucial to our being able achieve hardness in the regime of $\eps = \frac1{2}-o(1)$. While we cannot rule out the possibility that the previous reductions from {\sc Set Cover} and {\sc Vertex Cover} can be extended to this regime, we were unable to obtain such an extension despite our own best efforts---it seems that a fundamentally different approach is necessary.

In general, results basing the hardness of weak learning (of any learning task) on {\sl worst-case} complexity-theoretic assumptions remain relatively rare. One reason is because the setting of weak learning corresponds to that of average-case complexity, and so any such result will have to amplify worst-case hardness into average-case hardness within the confines of the learning task at hand.\footnote{While boosting establishes an equivalence of weak and strong learning, boosting algorithms do not preserve the structure of the hypothesis. For example, boosting a weak learner that returns a decision tree hypothesis yields a strong learner that returns a hypothesis that is the majority of decision trees. Therefore, the hardness of properly learning decision trees in the setting of strong learning does not immediately yield hardness in the setting of weak learning.}

\section{Discussion}

\paragraph{Two interpretations of our results.} The existing literature on properly learning decision trees is split roughly evenly between algorithms and hardness, and there is no consensus as to whether~\cite{EH89}'s algorithm is optimal. As for the approximability of $k$-NCP, there is a huge gap between the $O(n/\log n)$ ratio achieved by the algorithms of~\cite{BK02,APY09} and the constant-factor inapproximability results of~\cite{BELM16,Man20,LLL24}, and there is likewise no consensus as to what the optimal ratio is within this range.

\Cref{cor:improving EH gives new NCP algorithms} can be viewed either as a new avenue for designing approximation algorithms for $k$-NCP or as one for showing that~\cite{EH89}’s algorithm is optimal. With regards to the former perspective, as already mentioned~\cite{EH89}'s quasipolynomial-time algorithm has been improved for variants of the problem---for example, we have polynomial-time algorithms that return hypotheses that are slightly more complicated than decision trees~\cite{Bsh93} and almost-polynomial-time query algorithms for the uniform distribution~\cite{BLQT22}.  A  natural avenue for future work is to see if the ideas in these works can now be useful for $k$-NCP or its variants. 
As for the latter perspective, the $O(n/\log n)$-versus-constant gap in our understanding of the approximability of $k$-NCP is especially stark when compared to the {\sl unparameterized} setting, where NCP has long been known to be NP-hard to approximate to almost-polynomial ($n^{\Omega(1/\log\log n)}$) factors~\cite{DKS98,DKRS03}. We hope that our work provides additional motivation for getting lower bounds in the parameterized setting “caught up” with those in the unparameterized setting. %---indeed, improving the existing inappproximability ratio from constant to   $\Omega(\log n)$ will already close a longstanding open problem of learning theory. 

More broadly, recent years have seen a surge of progress on  parameterized inapproximability; see~\cite{FKLM20} for a survey.  Notably, for example, a recent breakthrough of Guruswami, Lin, Ren, Sun, Wu~\cite{GLRSW24} establishes the parameterized analogue of the PCP Theorem.\footnote{Their work also carries new implications for $k$-NCP, though the parameters achieved by~\cite{BELM16,Man20,LLL24} are quantitatively stronger for our purposes.} The framework of parameterized inapproximability syncs up especially nicely with the setup of learning theory: the parameterized setting is relevant because it allows us to control the size of the target function, and the inapproximability ratio corresponds to the gap in sizes between the target and hypothesis. We believe that there is much more to be gained, both in terms of algorithms and hardness, by further exploring connections between these two fields. 

\paragraph{Decision trees and weak learning in practice.} Our interest in the setting of weak learning is motivated in part by a specific use case of decision trees in practice. {\sl Tree ensemble methods} such as XGBoost~\cite{CG16} have emerged as powerful general-purpose  algorithms that achieve state-of-the-art performance across a number of settings (especially on tabular data where they often  outperform deep neural nets~\cite{SA22,GOV22}). Roughly speaking, these methods first construct an ensemble of decision trees, each of which is mildly correlated with the data, and then aggregate the predictions of these trees into an overall prediction. 

 Our results provide a theoretical counterpoint to the empirical success of these methods. We show that the task of finding even a single small single decision tree that is mildly correlated with the data---the task that is at the very heart of these ensemble methods---is intractable. Indeed,~\Cref{cor:W1 hardness of polytime,cor:explicit superpoly} show that this is the case even if the data is {\sl perfectly} labeled by a small decision tree---a strong stylized assumption that real-world datasets almost certainly do not satisfy. 

\paragraph{LPN hardness of uniform-distribution learning?} A criticism that can be levied against all existing lower bounds for properly learning decision trees, including ours, is that they only hold if the examples are distributed according to a worst-case distribution. It would therefore be interesting to establish the hardness of learning under ``nice" distributions, the most canonical one being the uniform distribution. Our work points to the possibility of basing such hardness on the well-studied {\sl Learning Parities with Noise} problem~\cite{BFKL93,BKW03} (LPN), a distributional variant of NCP where the input is a random linear code instead of a worst-case code. Unfortunately, our reduction does not preserve the uniformity of distributions---i.e.~it translates the hardness of LPN into the hardness of learning under a non-uniform distribution---but perhaps a modification of it can.

\section{Technical Overview for~\Cref{thm:most general}}

\subsection{Warmup: {\sc DT-Learn} solves  decisional approximate {\sc $k$-NCP}}
\label{sec:warmup}

We first show, as a warmup, how algorithms for {\sc DT-Learn} can be used to solve the {\sl decision version} of approximate $k$-NCP:

\begin{definition}[Decisional $\alpha$-approximate $k$-NCP]
\label{def:decisional NCP}
 Given as input the generator matrix $G \in \F_2^{n\times d}$ of a code $\mathcal{C}$, a received word $z\in \F_2^n$, a distance parameter $k\in \N$, and an approximation parameter $\alpha \ge 1$, distinguish between:     
\begin{itemize}
\item[$\circ$] \textnormal{Yes}: there is a codeword $y\in\mathcal{C}$ within Hamming distance $k$ of $z$;
\item[$\circ$] \textnormal{No}: the Hamming distance between $z$ and every codeword $y\in\mathcal{C}$  is greater than $\alpha k$.
\end{itemize}    
\end{definition}

\begin{theorem}[\Cref{thm:most general} for {\sl decisional} approximate $k$-NCP]
\label{thm:decision version of main result}
Suppose there is an algorithm that given random examples distributed according to a distribution $\mathcal{D}$ over $\zo^n$ and labeled by a size-$s$ decision tree, runs in time $t(n,s,s',\eps)$ and outputs a size-$s'$ decision tree hypothesis that achieves accuracy $1-\eps$ under $\mathcal{D}$. Then, for all $\ell\in \N$ there is an algorithm which solves \emph{decisional} $\alpha$-approximate $k$-\textnormal{NCP} running in time \[ O(\ell n^2)\cdot t(\ell n,2^{\ell k},2^{O(\alpha \ell k)},\eps)+\poly(n,\ell,2^{\alpha \ell k}) \text{\ where \ $\eps=\lfrac{1}{2}-2^{-\Omega(\alpha \ell k)}$.}\] 
\end{theorem}

There are no known search-to-decision reductions for approximate $k$-NCP, but in~\Cref{sec:search} we will explain how our proof of~\Cref{thm:decision version of main result} can be upgraded to show that algorithms for {\sc DT-Learn} in fact be used to solve the actual {\sl search} version of approximate $k$-NCP, thereby yielding~\Cref{thm:most general}.

\paragraph{Dual formulation.} 
 We begin by transforming~\Cref{def:decisional NCP} into its dual formulation where the algorithm is given as input the {\sl parity check matrix} of a code instead of its generator matrix:

\begin{definition}[Parity check view of decisional $\alpha$-approximate $k$-\textnormal{NCP}]
\label{def:parity-check-view}
    Given as input the parity check matrix $H \in \mathbb{F}_2^{m\times n}$ of a linear code and a target vector $t \in \mathbb{F}_2^m$, distinguish between: 
    \begin{itemize}
        \item[$\circ$] \textnormal{Yes}: there is a $k$-sparse vector $x\in\F_2^n$ such that $Hx=t$
        \item[$\circ$] \textnormal{No}: there does not exist a $\violet{\alpha k}$-sparse $x\in\F_2^n$ such that $Hx=t$.
    \end{itemize}
\end{definition}

This view of NCP is also known as {\sl syndrome decoding} in coding theory. The fact that one can efficiently switch between the two views of NCP is standard and follows by elementary linear algebra. The parity check view aligns especially well with the task of testing and learning an unknown function $f:\F_2^n\to\F_2$\footnote{For the rest of the paper, we switch to viewing Boolean functions as mapping vectors in $\F_2^n$ to $\F_2$ since this aligns well with the linear-algebraic nature of NCP and our proofs.} since it can be equivalently stated as follows.

\begin{definition}
\label{def:parity-consistent-view}
    Given as input a set $D = \{x^{(1)},\ldots,x^{(m)}\}\sse \F_2^n$ and a partial function $f: D \to \F_2$, distinguish between: 
        \begin{itemize}
        \item[$\circ$] \textnormal{Yes}: $f$ is a  $k$-parity  
        \item[$\circ$] \textnormal{No}: $f$ disagrees with every $\violet{\alpha k}$-parity on at least one input $x\in D$.
    \end{itemize}
\end{definition}

We have reformulated decisional $\alpha$-approximate $k$-\textnormal{NCP} as the problem of distinguishing between $f : \F_2^n \to \F_2$ being a $k$-parity under $\mathrm{Unif}(D)$ versus $\frac1{m}$-far from all $\violet{\alpha k}$-parities under $\mathrm{Unif}(D)$. 

\subsubsection{Our strategy}

Proving~\Cref{thm:decision version of main result} therefore amounts to amplifying the gap between the Yes and No cases in such a way that~$f$ remains a sparse parity in the Yes case, and yet becomes $(\frac1{2}-2^{-\Omega(\violet{\alpha k})})$-far from all decision trees of size $2^{\Omega(\violet{\alpha k})}$ in the No case.  We do so incrementally in three steps. See \Cref{fig:amplification} for an illustration \violet{of these steps and \Cref{fig:inclusions} for an illustration of the inclusions of the different function classes we consider}.

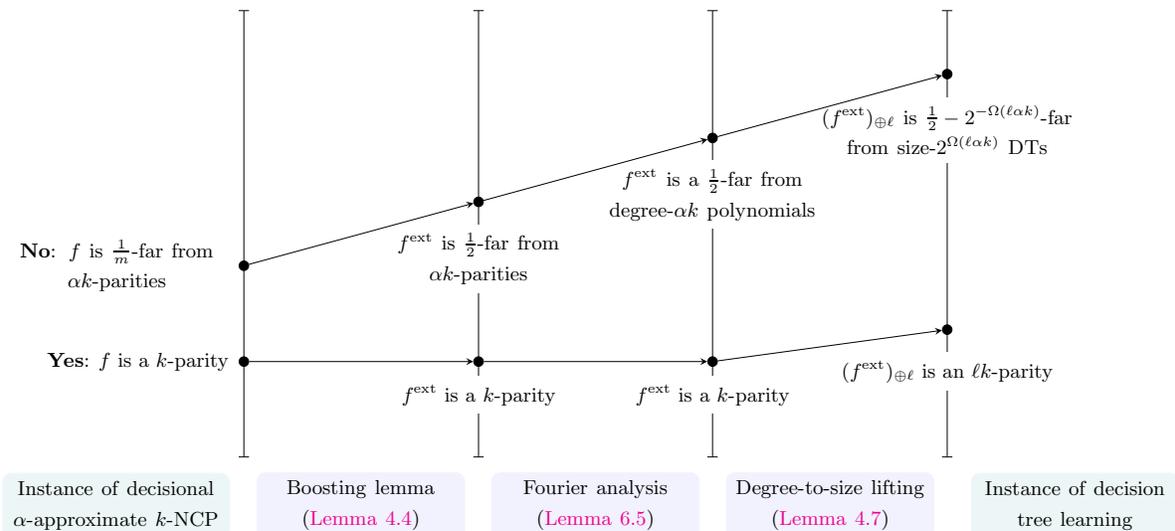
\begin{figure}[t]
    \centering
    \scalebox{0.85}{
    \begin{tikzpicture}[]
        % x DEFINES HOW WIDE FIGURE IS
        \def\x{5.5}
        % xh defines third width
        \def\xh{1.83}
        % xt defines two-thirds width
        \def\xt{3.67}
        % xf defines four-thirds width
        \def\xf{7.5}
        % c DEFINES CENTER
        \def\c{0}
        % l DEFINES LEFT CENTER
        \def\l{-0.1}
        % l DEFINES RIGHT CENTER
        \def\r{0.1}
        % DEFINES HOW LARGE THE NO CASE INCREMENT IS
        \def\inc{1}
        % DEFINES HOW FAR TO SHIFT LEFT CATEGORY LABEL FROM CENTER 
        \def\lef{-7}
        % DEFINES THE CATEGORY LABEL INCREMENT
        \def\cat{3.5}
        % DEFINES MAX HEIGHT
        \def\h{4}
        % Y AXIS STARTING POINT OF TOP AMPLIFICATION LINE
        \def\yt{0}
        % Y AXIS STARTING POINT OF BOTTOM AMPLIFICATION LINE
        \def\yb{-1.5}
        % TEXT WIDTH OF THEOREM LABELS
        \def\tw{3.0}
        % MAIN VERTICAL AXES
        \draw[black,xshift=-\x cm,|-|] (\c,-3) node[left] {{}} -- (\c,\h) node[left]{{}};
        \draw[black,xshift=-\xh cm,|-|] (\c,-3) node[left] {{}} -- (\c,\h) node[left]{{}};
        \draw[black,xshift=\xh cm,|-|] (\c,-3) node[left] {{}} -- (\c,\h) node[left]{{}};
        \draw[black,xshift=\x cm, |-|] (\c,-3) node[right] {{}} -- (\c,\h) node[right]{{}};
        % k AND k' LEFT
        \node[draw,circle,fill=black,inner sep=1.5pt,xshift=-\x cm] (k) at (\c,\yb) {};        
        \node[draw,circle,fill=black,inner sep=1.5pt,xshift=-\x cm] (k') at (\c,\yt) {};
        % k AND k' LEFT MIDDLE
        \node[draw,circle,fill=black,inner sep=1.5pt,xshift=-\xh cm] (k'LM) at ($(\c,\yt) + (0,\inc)$) {};
        \node[draw,circle,fill=black,inner sep=1.5pt,xshift=-\xh cm] (kLM) at (\c,\yb) {};

        % k AND k' RIGHT MIDDLE
        \node[draw,circle,fill=black,inner sep=1.5pt,xshift=\xh cm] (k'RM) at ($(\c,\yt) + (0,2*\inc)$) {};
        \node[draw,circle,fill=black,inner sep=1.5pt,xshift=\xh cm] (kRM) at (\c,\yb) {};

        % RIGHT AXES POINTS
        \node[draw,circle,fill=black,inner sep=1.5pt,xshift=\x cm] (k'R) at ($(\c,\yt) + (0,3*\inc)$) {};
        \node[draw,circle,fill=black,inner sep=1.5pt,xshift=\x cm] (kR) at (\c,-1) {};
        
        % LABELS FOR STEPS
        \draw[xshift=-\xf cm] (0,-3.75)  node [rounded corners, text centered, fill=teal!7, text width=3.3 cm] {\footnotesize {Instance of decisional \\ $\alpha$-approximate $k$-NCP} };
        \draw[xshift=-\xt cm] (0,-3.75)  node [text centered, text width=\tw cm, rounded corners,fill=blue!5] {\footnotesize {Boosting lemma} \\ (\Cref{lem:boosting})};
        \draw[xshift=0 cm] (0,-3.75) node [text centered, text width=\tw cm, rounded corners,fill=blue!5] {\footnotesize Fourier analysis \\ (\Cref{lem:parity-to-degree})};
        \draw[xshift=\xt cm] (0,-3.75) node [text centered, text width=\tw cm, rounded corners,fill=blue!5] {\footnotesize Degree-to-size lifting \\ (\Cref{lem:deg-to-size lifting})};
        \draw[xshift=\xf cm] (0,-3.75) node [text centered, text width=\tw cm, rounded corners,fill=teal!7] {\footnotesize Instance of decision \\ tree learning};
        
        % ALL LABELS
        % YES/NO LEFT LABELS
        \draw[color=black,xshift=-\x cm] (\l,\yt) node[text width=3.5cm,text centered, left,fill=white] {\footnotesize\color{black} \textbf{No}: $f$ is $\frac{1}{m}$-far from \\ $\alpha k$-parities};
        \draw[color=black,xshift=-\x cm] (\l,\yb) node[left,fill=white] {\footnotesize {\color{black}\textbf{Yes}: $f$ is a $k$-parity}};

         % LEFT MIDDLE LABELS
        \draw[] ([yshift=-.35cm,xshift=0cm]k'LM) node[anchor=north,text width=4cm,fill=white,text centered] {\footnotesize{$f^{\mathrm{ext}}$ is $\frac{1}{2}$-far from\\ $\alpha k$-parities}};
        \draw[] ([yshift=-.22cm,xshift=0cm]kLM) node[anchor=north,fill=white,text centered] {\footnotesize{$f^{\mathrm{ext}}$ is a $k$-parity}};

        % RIGHT MIDDLE LABELS
        \draw[] ([yshift=-.35cm,xshift=0cm]k'RM) node[anchor=north, fill=white,text centered, text width=4cm] {\footnotesize{$f^{\mathrm{ext}}$ is a $\frac{1}{2}$-far from degree-$\alpha k$ polynomials}};
        \draw[] ([yshift=-.22cm,xshift=0cm]kRM) node[anchor=north,fill=white] {\footnotesize{$f^{\mathrm{ext}}$ is a $k$-parity}};

        % RIGHT LABELS
        \draw[] ([yshift=-.35cm,xshift=0cm]k'R) node[anchor=north,fill=white,text centered,text width=4cm] {\footnotesize{$(f^{\mathrm{ext}})_{\oplus \ell}$ is $\frac{1}{2}-2^{-\Omega(\ell \alpha k)}$-far from size-$2^{\Omega(\ell \alpha k)}$ DTs}};   
        \draw[] ([yshift=-.35cm,xshift=0cm]kR) node[anchor=north,fill=white] {\footnotesize{$(f^{\mathrm{ext}})_{\oplus \ell}$ is an $\ell k$-parity}};
             
         % LINES
        \draw[black,-stealth] (kRM) to node[midway,below,sloped] {} (kR);
        \draw[black,-stealth] (k'RM) to node[midway,above,sloped] {} (k'R);
        \draw[black,-stealth] (k'LM) to node[midway,above,sloped] {} (k'RM);
        \draw[black,-stealth] (kLM) to node[midway,above,sloped] {} (kRM);
        \draw[black,-stealth] (k') to node[midway,above,sloped] {} (k'LM);
        \draw[black,-stealth] (k) to node[midway,below,sloped] {} (kLM);

    \end{tikzpicture}
    }
      \captionsetup{width=.9\linewidth}
    \caption{An illustration of~\Cref{thm:decision version of main result} as a series of gap amplification steps. Starting with an instance of $k$-NCP on the left, we perform a series of transformations to obtain an instance of the distinguishing problem on the right. 
    Due to space constraints we have omitted descriptions of the corresponding distributions from the figure. These distributions also go through a series of transformations, from $\mathrm{Unif}(D)$ on the left to $\mathrm{Unif}(\mathrm{Span}(D))_{\oplus \ell}$ on the right.} 
    \label{fig:amplification}
\end{figure}

%\lnote{Let's discuss the purple bubbles.}

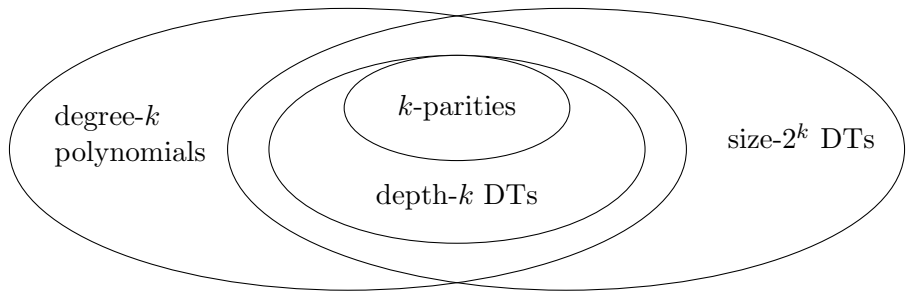
\begin{figure}[h!]
    \centering
    \begin{tikzpicture}[tips=proper]
        % nested ovals 
        \node[ellipse,draw = black,text = black,minimum width = 5cm,minimum height = 2.5cm,fill=white!10,label={[shift={(0,-2.2)}]{depth-$k$ DTs}}] (depth) at (0,-.5) {};
        % \node[ellipse,draw = black,text = black,minimum width = 4cm,minimum height = 1.8cm,fill=white!30,label={[shift={(0,-1.7)}]{$k$-juntas}}] (junta) at (0,-0.15) {};
        \node[ellipse,draw = black,text = black,minimum width = 3cm,minimum height = 1.4cm,fill=white!60] (parity) at (0,0.05) {$k$-parities};

        % non-nested ovals
        \node[ellipse,draw = black,text = black,minimum width = 9cm,minimum height = 3.75cm,label={[shift={(-2.7,-2.25)},text width=24mm]{degree-$k$\\ polynomials}}] (poly) at (-1.45,-.5) {};
        \node[ellipse,draw = black,text = black,minimum width = 9cm,minimum height = 3.75cm,label={[shift={(3.15,-2.0)},text width=20mm]{size-$2^k$ DTs}}] (size) at (1.45,-.5) {};
    \end{tikzpicture}
    \caption{Illustration of inclusions of basic function classes}
    \label{fig:inclusions}
\end{figure}

\paragraph{Step 1.} For the first step, we consider the linear span of $D$:
\[ \mathrm{Span}(D) \coloneqq \bigg\{ \sum_{i\in S} x^{(i)} \mid S \subseteq [m]\bigg\},
\] 
where we have assumed for simplicity that the vectors in $D$ are linearly independent. (Otherwise, the span is defined to be all possible linear combinations of the basis vectors of $D$.) We analogously consider $f$’s linear extension $f^{\mathrm{ext}} : \mathrm{Span}(D) \to \mathbb{F}_2$:  for all $S\sse [m]$, 
\[ f^{\mathrm{ext}}\bigg( \sum_{i\in S} x^{(i)}\bigg) = \sum_{i\in S} f\big(x^{(i)}\big)\]
and we prove the following ``boosting lemma": 

\begin{lemma}
\label{lem:boosting}
    For every set $D \sse \F_2^n$ and function $f: D\to\F_2$, we have: 
    \begin{itemize}
    \item[$\circ$] Preservation of the \textnormal{Yes} case: if $f$ is a parity $\chi_{S}$, then $f^{\mathrm{ext}}$ is also the parity $\chi_S$.
    \item[$\circ$] Amplification of the \textnormal{No} case: if $f$ disagrees with every $\violet{\alpha k}$-parity on at least one input in $D$, then $f^{\mathrm{ext}}$ disagrees with every $\violet{\alpha k}$-parity on exactly $\frac1{2}$ of the inputs in $\mathrm{Span}(D)$.
    \end{itemize}
\end{lemma}

Note that the domain of our function has been increased exponentially in size, since $|\mathrm{Span}(D)| = 2^{|\mathrm{dim}(D)|}$. Thankfully, this is not an issue since we will still be able to efficiently provide the learner with random examples sampled from this exponentially large set. 

%Similar boosting procedures were used in~\cite{FGKP09,KPV14,BGGS16}, but they were more elaborate since their applications required the size of the resulting set to be polynomial in the size of $D$. 

\paragraph{Step 2.}
The second step follows by Fourier analysis: if a function is uncorrelated with any small parity under $\mathcal{D}$, then by linearity of expectation, it is also uncorrelated with any low-degree Fourier polynomial under $\mathcal{D}$.

\paragraph{Step 3.} Finally, we give a generic way to lift lower bounds against low-degree polynomials to lower bounds against small-size decision trees. For intuition about this step, we briefly sketch an elementary proof for the case when $\mathcal{D}$ is the {\sl uniform} distribution. We claim that every small-size decision tree is well-approximated by a low-degree polynomial under the uniform distribution. To see this, note that truncating a size-$s$ tree $T$ at depth $d$ yields a tree $T_{\mathrm{trunc}}$ that is $(2^{-d}s)$-close to $T$ w.r.t.~the uniform distribution. This is because the fraction of inputs that follow any path of length $d$ is precisely $2^{-d}$ and we take a union bound over at most $s$ truncated paths. Finally, the fact that depth-$d$ decision trees have Fourier degree $d$ completes the proof. 

This proof fails for an arbitrary distribution $\mathcal{D}$ since the probability that a random $\bx\sim\mathcal{D}$ follows a path of length $d$ can now be much larger than $2^{-d}$. To overcome this, we show that by composing $\mathcal{D}$ with a parity gadget, it becomes “uniform enough” for this fact to hold. The parity gadget is defined as follows.

%\lnote{Use $g$ instead of $f$ below because people are expecting these to be applied to $f^{\mathrm{ext}}$.}

\paragraph{Notation.}  For a vector $y \in (\F_2^{\ell})^n$, we write $y^{(i)}\in \F_2^\ell$ to denote the $i$th block of $y$. We define the function $\BlockwisePar: (\F_2^\ell)^n \to \F_2^n$: 
\[ 
\BlockwisePar(y) \coloneqq (\oplus y^{(1)}, \ldots , \oplus y^{(n)}), 
\] 
where $\oplus y^{(i)}$ denotes the parity of the bits in $y^{(i)}$. 

\begin{definition}[Parity substitution in functions and distributions]
\label{def:parity-sub}
    For a function $g : \F_2^n\to\F_2$, the function $g_{\oplus \ell}:(\F_2^\ell)^n\to\F_2$ is defined as
    $$
    g_{\oplus \ell}(y)=g(\BlockwisePar(y)).
    $$
    For a distribution $\mathcal{D}$ over $\F_2^n$, the distribution $\mathcal{D}_{\oplus \ell}$ is defined via the following experiment:
    \begin{enumerate} 
    \item First sample $\bx \sim \mathcal{D}$. 
    \item For each $i\in [n]$, sample $\by^{(i)}\sim \F_2^\ell$ u.a.r.~among all strings satisfying $\oplus \by^{(i)} = \bx_i$. Equivalently, sample $\by\sim \mathcal{D}_{\oplus \ell}(\bx)$ where $\mathcal{D}_{\oplus \ell}(\bx)$ is the uniform distribution over all $y\in(\F_2^{\ell})^n$ satisfying $\BlockwisePar(y)=\bx$.
    \end{enumerate} 
\end{definition}

A key property of the parity substitution operation that for {\sl any} initial distribution $\mathcal{D}$, the parity-substituted distribution $\mathcal{D}_{\oplus \ell}$ becomes ``uniform-like'' in the sense that the probability a random $\by\sim \mathcal{D}_{\oplus \ell}$ is consistent with a fixed restriction decays exponentially in the length of the restriction.

\begin{proposition}[$\mathcal{D}_{\oplus \ell}$ is uniform-like]
    \label{prop:uniform-like-intro}
    For any $\ell \ge 2$, let $R\sse [n \ell]$ be a subset of coordinates and $r\in \F_2^{|R|}$. Then,
    $$
    \Prx_{\by\sim \mathcal{D}_{\oplus \ell}}[\by_R=r]\le 2^{-\Omega(|R|)}.
    $$
\end{proposition}

\Cref{prop:uniform-like-intro} together with a couple of additional observations yields: 

\begin{lemma}[Degree-to-size lifting]
\label{lem:deg-to-size lifting}
Let $\mathcal{D}$ be any distribution over $\F_2^n$ and suppose $g : \F_2^n \to \F_2$ is \violet{$\tfrac{1}{2}$}-far from all polynomials of Fourier degree $\violet{\alpha k}$ under~$\mathcal{D}$. Then for all $\ell \ge 2$, we have that $g_{\oplus \ell}:(\F_2^\ell)^n\to\F_2$ is $(\violet{\tfrac{1}{2}} - 2^{-\Omega(\violet{\ell \alpha k})})$-far from all decision trees of size $2^{O(\violet{\ell \alpha k})}$ under $\mathcal{D}_{\oplus \ell}$. 
\end{lemma}

\subsection{Proof of~\Cref{thm:most general}: {\sc DT-Learn} solves the search version of $k$-NCP} 
\label{sec:search}

As in the proof of~\Cref{thm:decision version of main result}, we first move from the generator matrix formulation of $k$-NCP to the parity check formulation (\Cref{def:parity-check-view}). We therefore assume that our input is of the form $(H,t) \in \F^{m\times n}_2 \times \F_2^m$ where there is a $k$-sparse vector $x \in \F^n_2$ such that $Hx = t$. Our goal, in the search version of approximate $k$-NCP, is to {\sl find} a $k'$-sparse vector $x' \in F^n_2$ such that $Hx' = t$, where~$k'$ is as close to $k$ as possible. By the equivalence between~\Cref{def:parity-check-view,def:parity-consistent-view}, this instance $(H,t)$ can be viewed as a set $D\sse \F^n_2$ and a $k$-parity $f : D \to \F_2$, and our goal can be equivalently stated as that of finding a $k'$-parity $h : D \to \F_2$ that agrees with $f$, where $k'$ is as close to $k$ possible. %\carnote{and also $h$ should agree with $f$ on $D$?}

Running through the 3-step transformation of the Yes case outlined in the previous section, we can efficiently provide the learner with random examples distributed according to $\mathrm{Unif}(\mathrm{Span}(D))_{\oplus \ell}$ and labeled by $(f^{\mathrm{ext}})_{\oplus \ell}$. Suppose the learner returns a size-$s'$ tree $T$ that is $\gamma$-correlated with $(f^{\mathrm{ext}})_{\oplus \ell}$ under $\mathrm{Unif}(\mathrm{Span}(D))_{\oplus \ell}$. We will show how the desired $k'$-parity $h : D \to \F_2$ can be extracted from $T$. Roughly speaking, this amounts to showing that the proof we sketched in the previous section can be ``unwound" to give an efficient algorithm for extracting such a parity. There are 4 steps to our analysis: 

\paragraph{Step 1.} By the contrapositive of~\Cref{claim:pruning-depth}, truncating $T$ at depth $\Theta(\log s') \eqqcolon k'$ yields a tree $T_{\mathrm{trunc}}$ that is 
$(\gamma - \Theta(\frac1{s'}))$-correlated with $(f^{\mathrm{ext}})_{\oplus \ell}$ under $\mathrm{Unif}(\mathrm{Span}(D))_{\oplus \ell}$. 

%\paragraph{Step 2.} Both $T_{\mathrm{trunc}}$ and $(f^{\mathrm{ext}})_{\oplus \ell}$ are functions over the domain $(\F_2^{\ell})^n$, enlarged from the original domain of $\F_2^n$ via the parity substitution operation (\Cref{def:parity-sub}):  
%\[ x_i \mapsto y_1^{(i)} \oplus \cdots\oplus y_{\ell}^{(i)}. \]
%We undo this operation by replacing each $y_j^{(i)}$-variable in $T_{\mathrm{trunc}}$ with the corresponding $x_i$-variable, resulting in a tree $T_{\mathrm{trunc}}^\star$ over $\F_2^n$ that we show is $(\gamma-\Theta(\frac1{s'}))$-correlated with $f^{\mathrm{ext}}$ under $\mathrm{Unif}(\mathrm{Span}(D))$. 

 \paragraph{Step 2.} Using basic Fourier-analytic properties of small-depth decision trees, we show that there exists a $k'$-parity $\chi_S$ in the Fourier support of $T_{\mathrm{trunc}}$ that is $((\gamma-\Theta(\frac1{s'}))4^{-k'})$-correlated with $(f^{\mathrm{ext}})_{\oplus \ell}$ under $\mathrm{Unif}(\mathrm{Span}(D))_{\oplus \ell}$.

\paragraph{Step 3.} Implicit in the proof of~\Cref{cor:parity-substitution} is that fact that we can undo the parity substitution operation and obtain from the aforementioned $k'$-parity $\chi_{S}$ a $(k'/\ell)$-parity $\chi_{S^\star}$ whose correlation with $f^{\mathrm{ext}}$ is the same as the correlation between $\chi_{S}$ and $(f^{\mathrm{ext}})_{\oplus \ell}$: 
\[ \Ex_{\bx\sim\mathcal{D}}\big[f^{\mathrm{ext}}(\bx)\chi_{S^\star}(\bx)\big] = \Ex_{\by\sim\mathcal{D}_{\oplus \ell}}\big[(f^{\mathrm{ext}})_{\oplus \ell}(\by)\chi_{S}(\by)\big] = \Big(\gamma-\Theta\Big(\frac1{s'}\Big)\Big)4^{-k'}.\]

 \paragraph{Step 4.}  Implicit in the proof of~\Cref{lem:boosting} is that fact that as long as the correlation between $\chi_{S^\star}$ and $f^{\mathrm{ext}}$ is positive, then $\chi_{S^\star}$ must in fact agree with $f^{\mathrm{ext}}$ on {\sl all} of $\mathrm{Span}(D)$, and hence with $f$ on all of $D$.

\section{Preliminaries} 

\paragraph{Notation and naming conventions.}{We write $[n]$ to denote the set $\{1,2,\ldots,n\}$. We use lower case letters to denote bitstrings e.g. $x,y\in\zo^n$ and subscripts to denote bit indices: $x_i$ for $i\in [n]$ is the $i$th index of $x$. For $R\sse [n]$, we write $x_R\in \zo^{|R|}$ to denote the substring of $x$ on the coordinates in $R$. A string $x\in \zo^n$ is $k$-sparse if it has at most $k$ nonzero entries. We use $\F_2$ to denote the finite field of order $2$. When dealing with finite fields, it will be convenient for us to identify a Boolean function on $n$ bits as a map $\F_2^n\to\F_2$.
}

\paragraph{Distributions.}{We use boldface letters e.g.~$\bx,\by$ to denote random variables. For a distribution $\mathcal{D}$, we write $\dist_{\mathcal{D}}(f,g)=\Pr_{\bx\sim \mathcal{D}}[f(\bx)\neq g(\bx)]$. A function $f$ is $\eps$-close to $g$ under $\mathcal{D}$ if $\dist_{\mathcal{D}}(f,g)\le \eps$. Similarly, $f$ is $\eps$-far from $g$ under $\mathcal{D}$ if $\dist_{\mathcal{D}}(f,g)\ge \eps$. { If $f$ is $0$-close under $\mathcal{D}$ to some $g$ having property $\mathcal{P}$, then we say that $f$ has property $\mathcal{P}$ under $\mathcal{D}$. For example, ``$f$ is a $k$-parity under $\mathcal{D}$'' means that there is a $k$-parity $g$ which is $0$-close to $f$ under $\mathcal{D}$.} For a set $S$, $\mathrm{Unif}(S)$ denotes the uniform distribution over that set. 
}

\paragraph{Parities and decision trees.}
For $S\sse[n]$, we write $\chi_S:\zo^n\to\zo$ to denote the parity of the coordinates in $S$. A $k$-parity function is a function $\chi_S$ for some $S\sse [n]$ with $|S|\le k$. A decision tree $T$ is a binary tree whose internal nodes query a coordinate and whose leaves are labeled by binary values. For a decision tree $T$, its size is the number of leaves in $T$ and is denoted $|T|$. 

% Its depth is the longest root-to-leaf path and is denoted $\depth(T)$. The decision tree $T$ on $n$ variables computes a function $T:\zo^n\to\zo$ where $T(x)$ is computed by starting at the root of $T$ and following the path consistent with the coordinates of $x$ and outputting the final leaf value.

% \paragraph{Fourier analysis.}{
% \lnote{Kinda weird to have a 1-liner about this. Maybe just move to the relevant spot in the body?} Each function $f:\F_2^n\to\R$ has a unique Fourier expansion
% $$
% f(x)=\sum_{S\sse [n]}\hat{f}(S)\chi_S(x)
% $$
% where $\hat{f}(S)\in\R$ is the Fourier coefficient of $f$ on $S$ and $\chi_S(x)=(-1)^{\sum_{i\in S}x_i}$ is the parity of the coordinates in $S$.
% }

\paragraph{Learning.}{In the PAC learning model, there is an unknown distribution $\mathcal{D}$ and some unknown \textit{target} function $f\in\mathcal{C}$ from a fixed \textit{concept} class $\mathcal{C}$ of functions over a fixed domain. An algorithm for learning $\mathcal{C}$ over $\mathcal{D}$ takes as input an error parameter $\eps\in (0,1)$ and has oracle access to an \textit{example oracle} $\textnormal{EX}(f,\mathcal{D})$. The algorithm can query the example oracle to receive a pair $(\bx,f(\bx))$ where $\bx\sim\mathcal{D}$ is drawn independently at random. The goal is to output a \textit{hypothesis} $h$ such that $\dist_{\mathcal{D}}(f,h)\le \eps$. Since the example oracle is inherently randomized, any learning algorithm is necessarily randomized. So we require the learner to succeed with some fixed probability e.g.~$2/3$. 
}

\subsection{Complexity-theoretic assumptions}
We list the hypotheses we use in order of strength of the hypothesis.
\begin{hypothesis}[{$\mathrm{W}[1]\neq \textnormal{FPT}$}, see \cite{DF13,CFKLMPPS15}]
For any computable function $\Phi:\N\to\N$, no algorithm can decide if a graph $G=(V,E)$ contains a $k$-clique in $\Phi(k)\cdot \poly(|V|)$ time. 
\end{hypothesis}

% Some of our results will be conditioned on the exponential time hypothesis.\lnote{Do we need the non-randomized versions? Also, let's order these according to strength.}
\begin{hypothesis}[Exponential time hypothesis (ETH) \cite{Tov84, IP01, IPZ01}]
There exists a constant $\delta>0$ such that $3$-SAT on $n$ variables cannot be solved in $O(2^{\delta n})$ time.
\end{hypothesis}

% Since we are proving hardness against learning and testing algorithms, which are inherently randomized, we will more specifically use a randomized variant of ETH.
% \begin{hypothesis}[Randomized ETH, see \cite{CIKP03,DHHTMTW14}]
% There exists a constant $\delta>0$ such that $3$-SAT on $n$ variables cannot be solved by a randomized algorithm in $O(2^{\delta n})$ time with error probability at most $1/3$.
% \end{hypothesis}

% We will also use two additional hypotheses.

\begin{hypothesis}[Gap-ETH \cite{D16,MR16}]
    There exist constants $\lambda,\delta>0$ such that no algorithm running in time $O(2^{\delta m})$ can solve the following task. Given a 3-SAT instance $\varphi$ with $m$ clauses distinguish between
    \begin{itemize}
        \item[$\circ$] \textnormal{Yes}: there exists an assignment of $\varphi$ satisfying all $m$ clauses; and
        \item[$\circ$] \textnormal{No}: every assignment of $\varphi$ satisfies at most $(1-\lambda)m$ clauses.
    \end{itemize}
\end{hypothesis}

% \lnote{It is weird for us to state both the deterministic and randomized versions of ETH, and then only the dterministic versions of Gap-ETH and $\mathrm{W}[1]\ne \mathrm{FPT}$. I propose that we just state the randomized versions of all of them -- in general, I don't think people are going to be too finicky about the distinction between the deterministic and randomized versions of these conjectures.} 

Our hardness results will be based on randomized versions of these hypotheses make the same runtime assumption but also against randomized algorithms. We remark that $\mathrm{W}[1]\neq \textnormal{FPT}$ is a weaker assumption than ETH which itself is weaker than Gap-ETH. %If $\mathrm{W}[1]=\textnormal{FPT}$, then SAT is solvable in subexponential time.

\subsection{Parameterized complexity of $k$-NCP}
\label{sec:hardness of NCP}
Bonnet, Egri, Lin, and Marx in \cite{BELM16} (see also~\cite{BBEGKLMM21}) show that obtaining any constant approximation of $k$-NCP is W[1]-hard:
\begin{theorem}[W{$[1]$}-hardness of approximating $k$-NCP, follows from {\cite[Theorem 2]{BELM16}}]
\label{thm:w1-hardness}
    Assuming $\textnormal{W[1]} \ne \textnormal{FPT}$, for all constants $c>1$, there is no algorithm running in time $\Phi(k)\cdot \poly(n)$ for any computable function $\Phi:\N\to\N$ that solves $c$-approximate $k$-\textnormal{NCP}.
\end{theorem}

Under ETH, a stronger hardness conjecture than $\mathrm{W}[1]\ne \mathrm{FPT}$, Li, Lin, and Liu \cite{LLL24} showed that a constant factor approximation is unattainable in time $n^{k^\delta}$ for constant $\delta>0$.

\begin{theorem}[ETH hardness of approximating $k$-NCP {\cite[Corollary 4]{LLL24}}]
\label{thm:eth-hardness}
    Assuming \textnormal{ETH}, for all constants $c>1$, there is no algorithm running in time $\Phi(k)\cdot n^{k^\delta}$ for any computable function $\Phi:\N\to\N$ and $\delta=\frac{1}{\polylog c}$ that solves $c$-approximate $k$-\textnormal{NCP}.
\end{theorem}

Under Gap-ETH, a stronger hardness conjecture than ETH, Manurangsi \cite{Man20} showed the same constant factor approximation is also unattainable even in time $n^{o(k)}$.
\begin{theorem}[Gap-ETH hardness of approximating $k$-NCP {\cite[Corollary 5]{Man20}}]
\label{thm:gap-eth-hardness}
    Assuming \textnormal{Gap-ETH}, for all constants $c>1$, there is no algorithm running in time $\Phi(k)\cdot n^{o(k)}$ for any computable function $\Phi:\N\to\N$ that solves $c$-approximate $k$-\textnormal{NCP}.
\end{theorem}

\section{{\sc DT-Learn} solves the decision version of $k$-NCP: Proof of \Cref{thm:decision version of main result}}

In this section, we prove the following from which \Cref{thm:decision version of main result} follows easily.

\begin{theorem}[Reducing decisional $k$-NCP to decision tree learning]
\label{thm:general main reduction}
    For all $\ell\ge 2$, the following holds. Given an instance $(G,z)$ of decisional $\alpha$-approximate $k$-\textnormal{NCP} over $\F_2^n$, there is function $g : (\F_2^{\ell})^n \to \F_2$ and a distribution $\mathcal{D}$ over $(\F_2^{\ell})^n$ such that the following holds.
\begin{enumerate}
 \item One can obtain random samples from $\mathcal{D}$ labeled by $g$ in $O(\ell n^2)$ time. 
\item If $(G,z)$ is a \textnormal{Yes} instance of decisional $\alpha$-approximate $k$-\textnormal{NCP} then $g$ is a $k\ell$-parity under $\mathcal{D}$. 
 \item If $(G,z)$ is a \textnormal{No} instance of decisional $\alpha$-approximate $k$-\textnormal{NCP} then $g$ is $(\frac1{2}-2^{-\Omega(\ell \alpha k)})$-far from every decision tree of size $2^{\Omega(\ell \alpha k)}$ under $\mathcal{D}$. 
\end{enumerate}
\end{theorem}

\subsection{Equivalent formulations of NCP}
%\paragraph{Equivalence of NCP variants.} 
\label{sec:equivalence of NCP variants}
In proving \Cref{thm:decision version of main result}, we will use the parity check view of NCP (\Cref{def:parity-check-view}). The fact that this formulation is equivalent to the generator view is standard and we include it here for completeness.

\begin{proposition}[Equivalence of the generator view and the parity check view of NCP]
\label{prop:equivalence of parity and generator views}
    The problem in \Cref{def:parity-check-view} is equivalent to $k$-\textnormal{NCP}. 
\end{proposition}

\begin{proof}
    Let $G\in \F_2^{n\times d}$ be the generator matrix for a code $\mathcal{C}$ and $z\in\F_2^n$, a received message. Let $H\in \F_2^{(n-d)\times n}$ be such that $H^{\top}$ is the generator of the dual code $\mathcal{C}^{\perp}$. The matrix $H$ can be efficiently computed from a generator matrix for the code $\mathcal{C}$. Furthermore, $H$ is the parity-check matrix for $\mathcal{C}$ since $Hx=0$ if and only if $x\in\mathcal{C}$. One can readily verify that the distance from $z$ to $\mathcal{C}$ is $k$ if and only if there is a $k$-sparse $x\in\F_2^n$ satisfying $Hx=Hz$. 
\end{proof}

The parity check view also lends itself nicely to being formulated as a learning task (\Cref{def:parity-consistent-view}). This fact is also standard and we include the equivalence for completeness.

\begin{proposition}[Equivalence of parity consistency problem and NCP]
    \label{prop:equivalence of parity consistency problem and NCP}
    The problem in \Cref{def:parity-consistent-view} is equivalent to the problem in \Cref{def:parity-check-view}.
\end{proposition}
\begin{proof}
    Let $H$ be the parity check matrix of a code $\mathcal{C}$ and $t\in\F_2^m$. The set $D=\{x^{(1)},\ldots,x^{(m)}\}$ consisting of the rows of $H$ and $f:D\to\F_2$ given by $f(x^{(i)})=t_i$ has the property that $Hx=t$ if and only if $x^{(i)}\cdot x=t_i$ for all $i=1,\ldots,m$. Therefore, if $x$ is $k$-sparse, then $f$ is a $k$-parity. Furthermore, if no $k'$-sparse $x$ satisfies $Hx=t$ then $f$ disagrees with every $k'$-parity on at least one point in $D$, and is therefore $\frac{1}{m}$-far from every such parity under $\mathrm{Unif}(D)$. 
\end{proof}

\begin{remark}[Linear independence of the vectors in $D$]
\label{remark:linear-independence}
    Implicit in the proof of \Cref{prop:equivalence of parity consistency problem and NCP} is the fact that the vectors in $D$ can be assumed to be linearly independent. This is because the parity check matrix $H$ is obtained by computing a basis (i.e. a set of linearly independent vectors) for the dual code $\mathcal{C}^{\top}$. This basis forms the rows of $H$ which are then used to form $D$. 
\end{remark}

With this view in hand, we proceed with the three main steps used to prove \Cref{thm:decision version of main result}.

\subsection{Step 1: The {Span} operation and its properties}
First, we show that we can efficiently generate random samples from the distribution $\mathrm{Unif}(\mathrm{Span}(D))$ labeled by $f^{\mathrm{ext}}$.
\begin{proposition}[Random samples from $\mathrm{Unif}(\mathrm{Span}(D))$ labeled by $f^{\mathrm{ext}}$]
\label{prop:sampling}
    Given a linearly independent set of vectors $D\sse \F_2^n$ and $f:D\to\F_2$, random examples from $\mathrm{Unif}(\mathrm{Span}(D))$ labeled by $f^{\mathrm{ext}}$ can be obtained in time $O(|D|n)$.
\end{proposition}

\begin{proof}
    Let $D=\{x^{(1)},\ldots,x^{(m)}\}$. Each $x\in\mathrm{Span}(D)$ can be written as a \textit{unique} sum $x=\sum_{i\in I}x^{(i)}$ for $I\sse[m]$. Therefore, to sample a pair $(\bx,f^{\mathrm{ext}}(\bx))$ where $\bx\sim \mathrm{Unif}(\mathrm{Span}(D))$ is uniform random, it is sufficient to sample a uniform random subset $\bI\sse[m]$ and return $(\sum_{i\in \bI}x^{(i)},\sum_{i\in \bI}f(x^{(i)}))$.
\end{proof}

\subsubsection{Proof of \Cref{lem:boosting}}
\paragraph{{Preservation of the Yes case}.}
Suppose that $f$ is the parity $\chi_S$. That is, for every $x\in D$, we have $\chi_S(x)=f(x)$. Then by linearity, we have for all $I\sse [m]$:
\begin{align*}
    \chi_S\left(\sum_{i\in I}x^{(i)}\right)&=\sum_{i\in I}\chi_S(x^{(i)})=\sum_{i\in I}f(x^{(i)}).
\end{align*}
This shows that $f^{\mathrm{ext}}:\mathrm{Span}(D)\to\F_2$ is the parity $\chi_S$.

\paragraph{{Amplification of the No case}.}
For the second point, let $\chi_S$ be a $k'$-parity for $k'=\alpha k$. Let $A\sse [m]$ indicate the set of points which are misclassified by $\chi_S$. That is, $i\in A$ if and only if $\chi_S(x^{(i)})\neq f(x^{(i)})$. Then, $\chi_S\left(\sum_{i\in I}x^{(i)}\right)=\mathrm{Parity}(|I\cap A|)+\sum_{i\in I}f(x^{(i)})$ which shows that
\begin{align*}
    \Prx_{\bI}\left[\chi_S\left(\sum_{i\in \bI}x^{(i)}\right)\neq \sum_{i\in \bI}f(x^{(i)})\right]&=\Prx_{\bI}\biggl[|\bI\cap A|\text{ is odd}\biggr]
\end{align*}
where $\bI\sse [m]$ is a uniform random subset of $[m]$.
Since $A\neq \varnothing$ by our assumption that any $k'$-parity disagrees with $f$ on at least one point, we have that $\Prx_{\bI}\biggl[|\bI\cap A|\text{ is odd}\biggr]=1/2$. Indeed, $\bI$ can equivalently be viewed as a uniform random string in $\bI\in \zo^{m}$ denoting the characteristic vector of the set. In this case, $|\bI\cap A|$ is odd if and only if the parity of the bits in the substring $\bI_{A}\in \zo^{|A|}$ is $1$ which happens with probability $1/2$ for a uniform random $\bI$. This shows that $\chi_S$ disagrees with $f^{\mathrm{ext}}$ on $1/2$ of the points in $\mathrm{Span}(D)$ as desired.\hfill\qed

\subsection{Step 2: Zero correlation with low-degree polynomials}

\begin{lemma}
    \label{lem:parity-to-degree}
    Let $g:\F_2^n\to\F_2$ be a function and $\mathcal{D}$ be a distribution over $\F_2^n$. If
    $$
    \dist_{\mathcal{D}}(g,\chi_S)=\lfrac{1}{2}\qquad \text{ for every $k'$ parity }\chi_S
    $$
    then,
    $$
    \dist_{\mathcal{D}}(g,h)=\lfrac{1}{2}\qquad \text{ for every }h\text{ with Fourier degree }\le k'.
    $$
\end{lemma}

\begin{proof}
    This proof uses basic Fourier analysis. As such, it will be convenient for us to regard $g:\F_2^n\to\F_2$ as a function $g:\F_2^n\to\R$ (this is achieved by mapping $\F_2$ to $\R$ via $0\to 1$ and $1\to -1$). The correlation of $g$ with any $k'$-parity $\chi_S$ under $\mathcal{D}$ is $0$ since
    \begin{align*}
        \Ex_{\bx\sim\mathcal{D}}[g(x)\chi_S(x)]&=\Prx_{\bx\sim\mathcal{D}}[g(\bx)=\chi_S(\bx)]-\Prx_{\bx\sim\mathcal{D}}[g(\bx)\neq \chi_S(\bx)]\\
        &=1-2\cdot\dist_{\mathcal{D}}(g,\chi_S)\\
        &=0.\tag{$\dist_{\mathcal{D}}(g,\chi_S)=\frac{1}{2}$}
    \end{align*}
    Therefore, the correlation under $\mathcal{D}$ between $g$ and any $h:\F_2^n\to\R$ whose polynomial degree is at most $k'$ is:
    \begin{align*}
        \Ex_{\bx\sim \mathcal{D}}[g(\bx)h(\bx)]&=\Ex_{\bx\sim \mathcal{D}}\Bigg[\Bigg(\sum_{{|S|\le k'}}\hat{h}(S)\chi_S(\bx)\Bigg)g(\bx)\Bigg]\\
        &=\sum_{|S|\le k'}\hat{h}(S)\Ex_{\bx\sim \mathcal{D}}[\chi_S(\bx)g(\bx)]\\
        &=0.
    \end{align*}
    This shows that $ \dist_{\mathcal{D}}(g,h)=\tfrac{1}{2}$ as desired.
\end{proof}

\subsection{Step 3: Proof of \Cref{lem:deg-to-size lifting}}
In this section, we prove \Cref{lem:deg-to-size lifting}. First, we establish some key properties of $f_{\oplus \ell}$ and $\mathcal{D}_{\oplus\ell}$ (recalling the relevant definitions from \Cref{def:parity-sub}).

\subsubsection{Properties of blockwise parity distribution}

If the distribution $\mathcal{D}$ can be efficiently sampled from, then so can the distribution $\mathcal{D}_{\oplus \ell}$. Likewise, if random samples from $\mathcal{D}$ can be labeled by $f$, then random samples from $\mathcal{D}_{\oplus \ell}$ can be labeled by $f_{\oplus\ell}$. This follows directly from the definition of parity substitution \Cref{def:parity-sub}.

\begin{fact}[Random samples from $\mathcal{D}_{\oplus \ell}$ labeled by $f_{\oplus\ell}$]
\label{fact:sampling-d-parity-ell}
    If there is a time-$t$ algorithm generating random samples from $\mathcal{D}$ labeled by $f:\F_2^n\to\F_2$, then there is an algorithm running in time $t+O(\ell n)$ for generating random samples from $\mathcal{D}_{\oplus \ell}$ labeled by $f_{\oplus\ell}$. 
\end{fact}

As mentioned in the introduction, a key property of the distribution $\mathcal{D}_{\oplus \ell}$ is that it is ``uniform-like'' in the sense that the probability a random $\by\sim \mathcal{D}_{\oplus \ell}$ is consistent with a fixed restriction decays exponentially in the length of the restriction.

\begin{proposition}[Formal version of~\Cref{prop:uniform-like-intro}]
    \label{prop:uniform-like}
    Let $R\sse [n \ell]$ be a subset of coordinates and $r\in \F_2^{|R|}$. Then,
    $$
    \Prx_{\by\sim \mathcal{D}_{\oplus \ell}}[\by_R=r]\le 2^{-|R|(1-1/\ell)}.
    $$
\end{proposition}

\begin{proof}
    For $i\in [n]$, let $R^{(i)}\sse [\ell]$ denote the $i$th block of $R$, that is the subset of coordinates of the $i$th block restricted by $R$. Let $r^{(i)}$ denote the corresponding substring of $r$ so that $r=(r^{(1)},\ldots,r^{(n)})$. We observe that for all $x\in \F_2^n$ for which $\Prx_{\by\sim \mathcal{D}_{\oplus \ell}(x)}[\by_{R^{(i)}}=r^{(i)}]$ is nonzero:
    $$
    \Prx_{\by\sim \mathcal{D}_{\oplus \ell}(x)}[\by_{R^{(i)}}=r^{(i)}]=\begin{cases}
        2^{-|R^{(i)}|} & |R^{(i)}|<\ell\\
        2^{-|R^{(i)}|+1} & |R^{(i)}|=\ell
    \end{cases}.
    $$
    If $|R^{(i)}|<\ell$, then the probability $\by_{R^{(i)}}=r^{(i)}$ is exactly $2^{-|R^{(i)}|}$: any subset of $\ell-1$ coordinates of the $i$th block of $\by$ is distributed uniformly at random. In the other case, $R^{(i)}$ consists of the entire $i$th block, in which case $\ell-1$ bits are distributed uniformly at random while the last bit is set according to $x$. In either case, we can write $\Prx_{\by\sim \mathcal{D}_{\oplus \ell}(x)}[\by_{R^{(i)}}=r^{(i)}]\le 2^{-|R^{(i)}|+|R^{(i)}|/\ell}$. Finally, we have
    \begin{align*}
        \Prx_{\by\sim \mathcal{D}_{\oplus \ell}}[\by_R=r]&=\Ex_{\bx\sim\mathcal{D}}\bigg[\Prx_{\by\sim \mathcal{D}_{\oplus \ell}(\bx)}[\by_R=r]\bigg]\\
        &=\Ex_{\bx\sim\mathcal{D}}\bigg[\prod_{i\in [n]}\Prx_{\by\sim \mathcal{D}_{\oplus \ell}(\bx)}[\by_{R^{(i)}}=r^{(i)}]\bigg]\tag{Independence of the blocks of $\by$ conditioned on $\bx$}\\
        &\le\prod_{i\in [n]} 2^{-|R^{(i)}|+|R^{(i)}|/\ell}\\
        &=2^{-|R|(1-1/\ell)}\tag{Definition of $R^{(i)}$}
    \end{align*}
    which completes the proof. 
\end{proof}

\subsubsection{A simple lemma about parity substitution}
For the next lemma, we switch to viewing a Boolean function as a mapping $g:\F_2^n\to\bits$.

\begin{lemma}
\label{lem:new parity amplification}
Let $g : \F_2^n \to \bits$ and $\mathcal{D}$ be a distribution over $\F_2^n$. Consider $g_{\oplus \ell} : (\F_2^\ell)^n\to\bits$ and $\mathcal{D}_{\oplus \ell}$.  We say that $S \sse [\ell n]$ is {\sl block-complete} if there is a set $S^\star\sse [n]$ such that $S$ contains all the coordinates in the blocks specified by $S^\star$ and no more. (This in particular implies that $|S^\star| = |S|/\ell$.) Then 
\[ \Prx_{\by\sim\mathcal{D}_{\oplus \ell}}[g_{\oplus\ell}(\by) = \chi_S(\by)] = 
\begin{cases}
\ds \Prx_{\bx\sim\mathcal{D}}[g(\bx) = \chi_{S^\star}(\bx)] & \text{if $S$ is block-complete} \\
\frac1{2} & \text{otherwise}.
\end{cases}
\]
\end{lemma}

\begin{proof}
 
First, suppose $S$ is block-complete. Then, the lemma follows simply by unpacking the definitions of $D_{\oplus \ell}$ and $g_{\oplus \ell}$. We will therefore assume that $S$ is not block-complete. 

Let $S^{(i)}$ be the intersection of $S$ and the $i$th block. Note that $S=\cup_{i=1}^n S^{(i)}$. For $i\in [n]$ and $x\in\F_2^n$, let  $\mathcal{D}_{\oplus\ell}^{(i)}(x)$ denote the distribution of $\by^{(i)}$ when $\by\sim\mathcal{D}_{\oplus\ell}(x)$. We make the following key observation: if there is an $i^*\in [n]$ such that $|S^{(i^*)}|<\ell$, then for every fixed $x$, $$\Ex_{\by^{(i^*)} \sim \mathcal{D}^{(i^*)}_{\oplus\ell}(x)}[ \chi_{S^{(i^*)}}(\by^{(i^*)})]=0.$$ This follows from the fact that the subset of $\by^{(i^*)}$ with indices in $S^{(i^*)}$ is a uniform random string, so its parity will be a uniform random bit. Note that such an $i^*$ exists if and only if $S$ is not block-complete. 
We will now show that $g_{\oplus \ell}(\by)$ and $ \chi_S(\by)$ have 0 correlation:
\begin{align*}
     \Ex_{\by \sim \mathcal{D}_{\oplus\ell}}[g_{\oplus \ell}(\by) \chi_S(\by)] &= \Ex_{\bx \sim \mathcal{D}}\bigg[\Ex_{\by \sim \mathcal{D}_{\oplus\ell}(\bx)}[g_{\oplus \ell}(\by) \chi_S(\by)]\bigg]\tag{Definition of $\mathcal{D}_{\oplus \ell}$}\\
     &=\Ex_{\bx \sim \mathcal{D}}\left[g(\bx) \Ex_{\by \sim \mathcal{D}_{\oplus\ell}(\bx)}[\chi_S(\by)]\right]\tag{Definition of $g_{\oplus \ell}$}\\
     &=\Ex_{\bx \sim \mathcal{D}}\left[g(\bx) \Ex_{\by \sim \mathcal{D}_{\oplus\ell}(\bx) }\left[\prod_{i=1}^n \chi_{S^{(i)}}(\by^{(i)})\right]\right]\tag{Definition of $S^{(i)}$}\\
      &=\Ex_{\bx \sim \mathcal{D}}\left[g(\bx) \prod_{i=1}^n \Ex_{ \by^{(i)} \sim \mathcal{D}^{(i)}_{\oplus\ell}(\bx)  }[\chi_{S^{(i)}}(\by^{(i)})]\right] \tag{Independence of $\by^{(i)}$ conditioned on $\bx$}\\
      &=0.\tag{Assumption that $S$ is not block-complete}
\end{align*}
The last equality follows from our key observation because $S$ is not block-complete, there is some $i^*\in [n]$ such that $|S^{(i^*)}|<\ell$. This shows that $\Prx_{\by\sim\mathcal{D}_{\oplus \ell}}[g_{\oplus\ell}(\by) = \chi_S(\by)]=\tfrac{1}{2}$ as desired.
%Let $\by' = \{\by^{(j)} | j \neq i\}$ and let $\mathcal{D}_{\oplus\ell}'(x)$ be the corresponding marginal distribution. 
    %\begin{align*}
        %\Prx_{\by \sim \mathcal{D}_{\oplus\ell}}[g_{\oplus \ell}(\by) = \chi_S(\by)] &= \Prx_{\bx \sim \mathcal{D}, \by' \sim \mathcal{D}'_{\oplus\ell}(x), \by^{(i)} \sim \mathcal{D}^{(i)}_{\oplus\ell}(x)}[g (\bx) = \chi_{S^{(i)}}(\by^{(i)})\chi_{S\setminus S^{(i)}}(\by')]\\
        %&= \Ex_{\bx \sim \mathcal{D},\by' \sim \mathcal{D}'_{\oplus\ell}(x)}\left[\Prx_{\by^{(i)} \sim \mathcal{D}^{(i)}_{\oplus\ell}(x)}[g (\bx) = \chi_{S^{(i)}}(\by^{(i)})\chi_{S\setminus S^{(i)}}(\by')]\right]. 
    %\end{align*}    
\end{proof}

\begin{corollary}
\label{cor:parity-substitution}
     If $g:\F_2^n\to\bits$ is $\frac1{2}$-far under $\mathcal{D}$ from all $k'$-parities, then for all $\ell\ge 1$, $g_{\oplus\ell}:(\F_2^\ell)^n\to\bits$ is $\frac1{2}$-far under $\mathcal{D}_{\oplus \ell}$ from every function of Fourier degree $\ell k'$.
\end{corollary}

\begin{proof}
    We observe that $g_{\oplus\ell}:(\F_2^\ell)^n\to\bits$ is $\frac1{2}$-far under $\mathcal{D}_{\oplus \ell}$ from $\ell k'$-parities. This is because, by \Cref{lem:new parity amplification}, for every $\ell k'$ parity $\chi_S$ there is a set $S^\star$ of size $\le k'$ such that: 
    \[ \Prx_{\by\sim\mathcal{D}_{\oplus \ell}}[g_{\oplus\ell}(\by) = \chi_S(\by)] = 
    \begin{cases}
    \ds \Prx_{\bx\sim\mathcal{D}}[g(\bx) = \chi_{S^\star}(\bx)]=\tfrac{1}{2} & \text{if $S$ is block-complete} \\
    \frac1{2} & \text{otherwise}
    \end{cases}
    \]
    where we used the assumption that $g$ is $1/2$-far under $\mathcal{D}$ from all $k'$-parities. The corollary then follows directly from \Cref{lem:parity-to-degree}.
\end{proof}

\subsubsection{Proof of~\Cref{lem:deg-to-size lifting}}

We now prove the main lemma showing that parity substitution lifts decision tree depth lower bounds to size lower bounds. 

\begin{lemma}[Generalization of~\Cref{lem:deg-to-size lifting}]
\label{lem:parity-substitution-technical}
    Let $\mathcal{D}$ be a distribution over $\F_2^n$ and $g:\F_2^n\to\F_2$. For every $\ell\ge 2$, the distribution $\mathcal{D}_{\oplus \ell}$ and the function $g_{\oplus \ell}:(\F_2^\ell)^n\to \F_2$ satisfy the following:
    \begin{enumerate}
        \item If $g$ is a $k$-parity under $\mathcal{D}$, then $g_{\oplus \ell}$ is a $k\ell$-parity under $\mathcal{D}_{\oplus \ell}$
        % \item If $f$ is $\eps$-far from every decision tree of expected depth $k'\ell$ under $\mathcal{D}$, then $f_{\oplus \ell}$ is $(\eps-2^{-\Theta(\ell k')})$-far from every decision tree of size $2^{\Theta(\ell k')}$.
        \item If $g$ is $\tfrac{1}{2}$-far under $\mathcal{D}$ from every degree-$k'$ polynomial, then $g_{\oplus \ell}$ is $(\frac{1}{2}-2^{-\ell k'/6})$-far under $\mathcal{D}_{\oplus \ell}$ from every decision tree of size $2^{\ell k'/3}$.
        % \begin{enumerate}
        %     \item $f_{\oplus \ell}$ is $\eps$-far under $\mathcal{D}_{\oplus \ell}$ from every degree-$\ell k'$ polynomial; and
        %     \item $f_{\oplus \ell}$ is $(\eps-2^{-\ell k'/6})$-far $\mathcal{D}_{\oplus \ell}$ from every decision tree of size $2^{\ell k'/3}$
        % \end{enumerate}
    \end{enumerate}
\end{lemma}

The proof of \Cref{lem:parity-substitution-technical} uses the following claim.

\begin{claim}[Pruning the depth of a decision tree]
    \label{claim:pruning-depth}
    Let $T$ be a size-$s$ decision tree and $c\in\N$ a parameter. Let $T'$ be the decision tree obtained from $T$ by pruning each path at depth $c\log (s)$. Then, for all $\ell\ge 1$, $\dist_{\mathcal{D}_{\oplus \ell}}(T',g_{\oplus \ell})\le\dist_{\mathcal{D}_{\oplus \ell}}(T,g_{\oplus \ell})+s^{1-c(1-1/\ell)}$. 
\end{claim}

\begin{proof}
    Let $\Pi$ denote the set of paths in $T$ which have been pruned. The size of $\Pi$ is at most $s$. First, we bound the probability that a random input follows a path in $\Pi$:
    \begin{align*}
        \Prx_{\by \sim \mathcal{D}_{\oplus \ell}}[\by\text{ follows a path in }\Pi]&\le \sum_{\pi \in \Pi}\Prx_{\by \sim \mathcal{D}_{\oplus \ell}}[\by\text{ follows }\pi]\tag{Union bound}\\
        &\le \sum_{\pi\in \Pi}2^{-c\log(s)(1-1/\ell)}\tag{\Cref{prop:uniform-like} and $|\pi|\ge c\log(s)$}\\
        &\le s^{1-c(1-1/\ell)}\tag{$|\Pi|\le s$}.
    \end{align*}
    Therefore:
    \begin{align*}
        \dist_{\mathcal{D}_{\oplus \ell}}(T',g_{\oplus \ell})&\le\dist_{\mathcal{D}_{\oplus \ell}}(T,g_{\oplus \ell})+\Prx_{\by \sim \mathcal{D}_{\oplus \ell}}[\by\text{ follows a path in }\Pi]\tag{Union bound}\\
        &\le \dist_{\mathcal{D}_{\oplus \ell}}(T,g_{\oplus \ell})+s^{1-c(1-1/\ell)}
    \end{align*}
    which completes the proof.
\end{proof}

\begin{proof}[Proof of \Cref{lem:parity-substitution-technical}]
    We prove each point separately.\bigskip

    \noindent
    1. Let $S\sse [n]$ denote the $k$ indices of the parity consistent with $g$ under $\mathcal{D}$. Then,
        \begin{align*}
            g_{\oplus \ell}(y)=g(\mathrm{BlockwisePar}(y))=\bigoplus_{i\in S}\oplus y^{(i)}
        \end{align*}
        is a $k\ell$-parity under $\mathcal{D}_{\oplus \ell}$.\bigskip

    \noindent
    2. We prove this statement by contradiction. Let $T$ be a decision tree of size $2^{\ell k'/3}$ achieving small error: $\dist_{\mathcal{D}_{\oplus \ell}}(T,g_{\oplus\ell})<\frac{1}{2}-2^{-\ell k'/6}$. Let $T'$ be the decision tree obtained by pruning each path of $T$ at depth $\ell k'$. Then,
        \begin{align*}
            \dist_{\mathcal{D}_{\oplus \ell}}(T',g_{\oplus \ell})&\le \dist_{\mathcal{D}_{\oplus\ell}}(T,g_{\oplus\ell})+(2^{\ell k'/3})^{1-3(1-1/\ell)}\tag{\Cref{claim:pruning-depth}}\\
            &\le\dist_{\mathcal{D}_{\oplus\ell}}(T,g_{\oplus\ell})+2^{-\ell k'/6}\tag{$\ell\ge 2$}\\
            &<\lfrac1{2}.\tag{$\dist_{\mathcal{D}_{\oplus \ell}}(T,g_{\oplus\ell})< \frac{1}{2}-2^{-\ell k'/6}$}
        \end{align*}
        Since $T'$ is a decision tree of depth $\ell k'$, it is a polynomial of degree $\ell k'$. However, since $g$ is $\frac{1}{2}$-far from polynomials of degree $k'$, we know that $g_{\oplus \ell}$ is $\frac{1}{2}$-far from polynomials of degree $\ell k'$ by \Cref{cor:parity-substitution}. Therefore, we have reached a contradiction and conclude that $g_{\oplus \ell}$ must be $(\frac{1}{2}-2^{-\ell k'/6})$-far from decision trees of size $2^{\ell k'/3}$. 
        % Finally, for all $j\in [\ell]$,
        % \begin{align*}
        %     \eps&\ge \dist_{D_{\oplus\ell}}(T',f_{\oplus\ell})\\
        %     &=\Ex_{\bz\sim\mathrm{Unif}(\zo^{(\ell-1)n})}\bigg[\dist_{\mathcal{D}}(T'(\mathrm{ParComplete}_j(\bz,\cdot)),f)\bigg]
        % \end{align*}
        % which implies there is a fixed string $z\in\F_2^{n(\ell-1)}$ such that $\eps\ge \dist_{\mathcal{D}}(T'(\mathrm{ParComplete}_j(\bz,\cdot)),f)$. The function $x\mapsto T'(\mathrm{ParComplete}_j(z,x))$ is computed by the decision tree $T'$ restricted according to $z$. This restricted tree has depth at most $\mathrm{depth}(T')\le k'$. Recalling the inclusion $\{\text{depth-}k' \text{ DTs}\}\sse\{\text{degree-}k' \text{ polynomials}\}$ (see \Cref{fig:inclusions}), $T'$ is also a polynomial of degree $k'$. But this contradicts our assumption that $f$ is $\eps$-far from all polynomials of degree $k'$ under $\mathcal{D}$.
\end{proof}

\subsection{Putting things together: Proof of \Cref{thm:general main reduction}}
Let $(H,t)\in\F_2^{m\times n}\times\F_2^m$ be an instance of decisional $\alpha$-approximate $k$-NCP where $H$ is the parity check matrix for the code $\mathcal{C}$. Let $D=\{x^{(1)},\ldots,x^{(m)}\}\sse \F_2^n$ be the set corresponding to the rows of the parity check matrix $H$ and $f:D\to\F_2$ be the function labeling the set according to $t$, $f(x^{(i)})=t_i$. Let $\mathcal{D}$ be the distribution $\mathrm{Unif}(\mathrm{Span}(D))_{\oplus \ell}$. That is, $\mathcal{D}$ is the distribution obtained by substituting a parity of size $\ell$ into $\mathrm{Unif}(\mathrm{Span}(D))$. Let $f^{\mathrm{ext}}:\mathrm{Span}(D)\to\F_2$ be the linear extension of $f$ to $\mathrm{Span}(D)$. We prove the theorem for the function $(f^{\mathrm{ext}})_{\oplus \ell}$. We split into cases. 

\paragraph{Yes case: there is a $k$-sparse vector $x$ such that $Hx=t$.}{
    We obtain the desired result from the following chain of observations
    \begin{enumerate}[itemsep=-0ex]
        \item $f:D\to \F_2$ is a $k$-parity (assumption of the $\textsc{Yes}$ case and \Cref{def:parity-consistent-view})
        \item ...which implies $f^{\mathrm{ext}}$ is a $k$-parity under $\mathrm{Unif}(\mathrm{Span}(D))$ (\Cref{lem:boosting})
        \item ...which implies $(f^{\mathrm{ext}})_{\oplus \ell}$ is a $\ell k$-parity under $\mathrm{Unif}(\mathrm{Span}(D))_{\oplus \ell}$ (\Cref{lem:parity-substitution-technical}).
    \end{enumerate}
}

\paragraph{No case: $Hx\neq t$ for all vectors $x$ of sparsity at most $\alpha k$.}{
    In this case, we make the following observations
    \begin{enumerate}[itemsep=-0ex]
        \item $f:D\to \F_2$ is disagrees with every $\alpha k$-parity on some $x\in D$ (assumption of the $\textsc{No}$ case and \Cref{def:parity-consistent-view})
        \item ...which implies that $\dist_{\mathrm{Unif}(\mathrm{Span}(D))}(f^{\mathrm{ext}},\chi_S)=\frac{1}{2}$ with every $\alpha k$-parity (\Cref{lem:boosting})
        \item ...which implies $(f^{\mathrm{ext}})_{\oplus\ell}$ is $\frac{1}{2}$-far from every function of Fourier degree at most $\ell \alpha k$ under $\mathrm{Unif}(\mathrm{Span}(D))_{\oplus \ell}$ (\Cref{cor:parity-substitution})
        \item ...which implies $(f^{\mathrm{ext}})_{\oplus \ell}$ is $(1/2-2^{-\Omega(\alpha\ell k)})$-far under $\mathrm{Unif}(\mathrm{Span}(D))_{\oplus \ell}$ from every decision tree of size $2^{O(\alpha \ell k)}$. (\Cref{lem:parity-substitution-technical})
    \end{enumerate}
}
Finally, we remark that by \Cref{prop:sampling}, random samples from $\mathrm{Unif}(\mathrm{Span}(D))$ labeled by $f^{\mathrm{ext}}$ can be efficiently generated and therefore by \Cref{fact:sampling-d-parity-ell}, so can random samples from $\mathrm{Unif}(\mathrm{Span}(D))_{\oplus \ell}$ labeled by $(f^{\mathrm{ext}})_{\oplus \ell}$.\hfill $\qed$

\subsection{Proof of \Cref{thm:decision version of main result}}
Let $\mathcal{A}$ be the decision tree learning algorithm from the theorem statement.

\paragraph{The reduction.}{
        Let $(H,t)\in\F_2^{m\times n}\times\F_2^m$ be an instance of decisional $\alpha$-approximate $k$-NCP where $H$ is the parity check matrix for the code $\mathcal{C}$. Using \Cref{thm:general main reduction}, we obtain a function $g : (\F_2^{\ell})^n \to \F_2$ and a distribution $\mathcal{D}$ over $(\F_2^{\ell})^n$. We run the algorithm $\mathcal{A}$ on $g$ and $\mathcal{D}$ for $t(\ell n, 2^{\ell k}, 2^{O(\alpha\ell k)}, \eps)$ time steps for $\eps=\tfrac{1}{2}-2^{-\Omega(\alpha \ell k)}$. Let $T_{\mathrm{hyp}}$ be the decision tree learned by $\mathcal{A}$. We compute an estimate, $\overline{\varepsilon}$, of the quantity $\dist_{\mathcal{D}}(g,T_{\mathrm{hyp}})$ to accuracy $\pm 2^{-\Omega(\alpha \ell k)}$ using an additional $2^{O(\alpha \ell k)}$ samples from $\mathcal{D}$ labeled by $g$. We return ``Yes'' if $\overline{\varepsilon}\le \frac{1}{2}-2^{-\Omega(\alpha \ell k)}$ and $|T_{\mathrm{hyp}}|\le 2^{O(\alpha\ell k)}$, and ``No'' otherwise.\footnote{Concretely, the constants hidden by the big-O notation are the following. If $\beta=1/2-2^{-\Omega(\alpha k \ell)}$ is the the error in the No case of \Cref{thm:general main reduction}, we require the learner output a hypothesis with error $\eps=\tfrac{1}{2}-2^{-c\alpha \ell k}$ where $c$ is a constant chosen so that $\eps<\beta$. Then, we estimate $\dist_{\mathcal{D}}(g,T_{\mathrm{hyp}})$ to accuracy $\pm 2^{-C\alpha \ell k}$ where $C$ is a large enough constant such that $\eps+2^{-C\alpha \ell k}<\beta$. Finally, we ``Yes'' if and only if $\overline{\eps}\le \eps+2^{-C\alpha \ell k}$ and $|T_{\mathrm{hyp}}|\le 2^{O(\alpha\ell k)}$.}
    }

\paragraph{Runtime.}{
    Random samples from $\mathcal{D}$ labeled by $g$ can be obtained in $O(\ell n^2)$-time. We simulate $\mathcal{A}$ for $t(\ell n, 2^{\ell k}, 2^{O(\alpha\ell k)}, \eps)$ time steps and estimating $\overline{\varepsilon}$ takes time $\poly(n,\ell,2^{\alpha \ell k})$. So the overall runtime of the reduction is $O(\ell n^2)\cdot t(\ell n, 2^{\ell k}, 2^{O(\alpha\ell k)}, \eps)+\poly(n,\ell,2^{\alpha \ell k})$.
}

\paragraph{Correctness.}{
        To prove correctness, we show that if $(H,t)$ is a Yes instance of decisional $\alpha$-approximate $k$-NCP, then we output Yes with high probability, and otherwise if $(H,t)$ is a No instance then our algorithm outputs No with high probability.\bigskip

{\bf Yes case: there is a $k$ sparse vector $x$ such that $Hx=t$.} In this case, $g$ is a parity of at most $k\ell$ variables under $\mathcal{D}$ by \Cref{thm:general main reduction}. Therefore, $g$ is a decision tree of size $2^{k\ell}$ under $\mathcal{D}$. By running $\mathcal{A}$ for $t(\ell n, 2^{\ell k}, 2^{O(\alpha\ell k)}, \eps)$ time steps, we obtain a decision tree $T_{\mathrm{hyp}}$ of size $|T_{\mathrm{hyp}}|\le 2^{O(\alpha\ell k)}$ which satisfies
$$
\dist_{\mathcal{D}}(g,T_{\mathrm{hyp}})\le\eps=\frac{1}{2}-2^{-\Omega(\alpha \ell k)}
$$
and therefore our estimate $\overline{\varepsilon}$ satisfies
$$
\overline{\varepsilon}\le \dist_{\mathcal{D}}(g,T_{\mathrm{hyp}})+2^{-\Omega(\alpha \ell k)}\le \frac{1}{2}-2^{-\Omega(\alpha \ell k)}
$$
with high probability which ensures that our algorithm correctly outputs ``Yes.''\bigskip 
        
 {\bf No case: $Hx\neq t$ for all vectors $x$ of sparsity at most $\alpha k$.} 
First, if $T_{\mathrm{hyp}}$ does not satisfy $|T_{\mathrm{hyp}}|\le 2^{O(\alpha\ell k)}$ then our algorithm correctly outputs ``No''. Otherwise, assume that $|T_{\mathrm{hyp}}|\le 2^{O(\alpha\ell k)}$. We will show that $T_{\mathrm{hyp}}$ must have large error so that in this case our algorithm also correctly outputs ``No''. 

\Cref{thm:general main reduction} implies that $g$ is $\frac{1}{2}-2^{-\Omega(\alpha k\ell)}$ far under $\mathcal{D}$ from every decision tree of size $2^{O(\alpha k\ell)}$. This implies that $\dist_{\mathcal{D}}(g,T_{\mathrm{hyp}})>\frac{1}{2}-2^{-\Omega(\alpha \ell k)}$. Therefore, our estimate $\overline{\eps}$ satisfies
$$
\overline{\eps}\ge \dist_{\mathcal{D}}(g,T_{\mathrm{hyp}})-2^{-\Omega(\alpha \ell k)}>\frac{1}{2}-2^{-\Omega(\alpha \ell k)}
$$
with high probability. This ensures that our algorithm correctly outputs ``No''.}\hfill\qed
\section{{\sc DT-Learn} solves the search version of $k$-NCP: Proof of~\Cref{thm:most general}}

\begin{claim}[Solving the search version of $k$-NCP given a decision tree]
\label{claim:search}
    Let $(H,t)\in \F_2^{m\times n}\times\F_2^m$ be an instance of \textnormal{NCP} where $H$ is the parity check matrix for the linear code, and let $D=\{y^{(1)},\ldots,y^{(m)}\}$ be the set corresponding to the rows of the parity check matrix $H$. 
    
    Let $f:D\to\F_2$ be the function satisfying $f(y^{(i)})=t_i$ for $i\in [m]$, $\mathcal{D}$ be the distribution $\mathrm{Unif}(\mathrm{Span}({D}))_{\oplus \ell}$, and $T$ be a size-$s$ decision tree satisfying $\dist_{\mathcal{D}}(T,(f^{\mathrm{ext}})_{\oplus \ell})\le \frac{1}{2}-\gamma$ where $\gamma\ge\Omega(s^{1-c(1-1/\ell)})$ for some $c\in \N$. 
    
    There is an algorithm running in time $\poly(n,\ell,1/\gamma^2,s)$ which outputs with high probability a set of coordinates $S\sse [n]$ such that $|S|\le \tfrac{c\log s}{\ell}$ and $\chi_S(y)=f(y)$ for all $y\in D$.
\end{claim}

Before proving the claim, we prove two helpful lemmas.

\begin{lemma}[Extracting a well-correlated parity from a decision tree]
\label{lem:extraction}
    Let $T$ be a depth-$d$ decision tree satisfying $\dist_{\mathcal{D}}(T,g)\le \frac{1}{2}-\gamma$ for some $\gamma>0$, distribution $\mathcal{D}$ over $\F_2^n$, and $g:\F_2^n\to\F_2$. Then, there is a $\poly(n,1/\gamma^2,2^d)$-time algorithm which uses $2^{O(d)}/\gamma^2$ random samples from $\mathcal{D}$ labeled by $g$ and with high probability outputs set of coordinates $S\sse [n]$ such that $|S|\le d$ and $\dist_{\mathcal{D}}(\chi_S,g)\le \frac{1}{2}-\Theta(\gamma 4^{-d})$.
\end{lemma}

The proof of \Cref{lem:extraction} relies on the following properties of the Fourier spectrum of decision trees.

\begin{fact}[Fourier spectrum of decision trees]
\label{fact:dt-fourier}
    Let $T$ be a depth-$d$ decision tree on $n$ variables. Then, the following properties hold.
    \begin{enumerate}
        \item If a Fourier coefficient of $T$, $\hat{T}(S)$, for $S\sse [n]$ is nonzero then $S$ consists of coordinates queried along some path in $T$.
        \item The number of nonzero Fourier coefficients is at most $4^d$.
        % \item Each nonzero Fourier coefficient, $\hat{T}(S)$, satisfies $|\hat{T}(S)|\ge 2^{-d}$.
    \end{enumerate}
\end{fact}

Property 2 in \Cref{fact:dt-fourier} follows immediately from property 1. A good reference for these properties is \cite[Section 3.2]{ODBook}.

\begin{proof}[Proof of \Cref{lem:extraction}]
    We show the following algorithm proves the lemma:
    \begin{enumerate}
        \item Draw $2^{O(d)}/\gamma^2$ random samples from $\mathcal{D}$ labeled by $g$.
        \item For every $S\sse[n]$ consisting of coordinates queried along some path in $T$, use the random samples to estimate $\Ex_{\bx\sim\mathcal{D}}[\chi_S(\bx)g(\bx)]$.
        \item Output the subset $S$ corresponding to the parity $\chi_S$ which is most well-correlated with $g$ over $\mathcal{D}$.
    \end{enumerate}
    There are at most $4^d$ subsets $S\sse[n]$ to check in step (2). Therefore, the runtime of this algorithm is $\poly(n,1/\gamma^2,2^d)$. It remains to prove correctness.
    
    Rewriting the assumption that $\Prx_{\bx\sim \mathcal{D}}[T(\bx)\neq g(\bx)]\le \frac{1}{2}-\gamma$ in terms of correlation, we have
    \begin{align*}
        \gamma&\le\Ex_{\bx\sim\mathcal{D}}[T(\bx)g(\bx)]\\
        &\le \sum_{S\sse [n]}\Ex_{\bx\sim\mathcal{D}}[\hat{T}(S)\chi_S(\bx)g(\bx)].\tag{Fourier expansion of $T$}
    \end{align*}
    By \Cref{fact:dt-fourier}, the number of nonzero Fourier coefficients $\hat{T}(S)$ is at most $4^d$ and therefore, there is some $S\sse [n]$ such that
    \begin{align*}
        \frac{\gamma}{4^d}&\le \Ex_{\bx\sim\mathcal{D}}[\hat{T}(S)\chi_S(\bx)g(\bx)]\\
        &\le \Ex_{\bx\sim\mathcal{D}}[\chi_S(\bx)g(\bx)].\tag{$\hat{T}(S)\le 1$}
    \end{align*}
    Moreover, this $S$ consists of coordinates queried along some path in $T$ by \Cref{fact:dt-fourier}. Using $2^{O(d)}/\gamma^2$ random samples from $\mathcal{D}$ labeled by $g$, the correlation $\Ex_{\bx\sim\mathcal{D}}[\chi_S(\bx)g(\bx)]$ can be estimated to within an additive accuracy of $\Theta(\frac{\gamma}{4^d})$ with a failure probability of $2^{-\Theta(d)}$. By a union bound over all $4^d$ subsets $S\sse[n]$ that the algorithm checks, all correlation estimates are within the desired accuracy bounds, and the algorithm successfully outputs a parity which achieves accuracy $1/2+\Theta(\gamma 4^{-d})$ in approximating $g$ over $\mathcal{D}$.\end{proof}
    % , estimate the correlation $\Ex_{\bx\sim\mathcal{D}}[\chi_S(\bx)f(\bx)]$, and output the most well-correlated $S$.  
    
    % Using $2^{O(d)}/\gamma^2$ random samples from $\mathcal{D}$ labeled by $f$, each $\Ex_{\bx\sim\mathcal{D}}[\chi_S(\bx)f(\bx)]$ can with probability $1-2^{\Theta(d)}$ be estimated to within an additive accuracy of $\Theta(\frac{\gamma}{4^d})$

    % $$
    % \frac{\gamma}{4^d}\le \Ex_{\bx\sim\mathcal{D}}[\hat{T}(S)\chi_S(\bx)f(\bx)]
    % $$
    % for some $S\sse [n]$. 

    % &\le \sum_{S\sse [n]}\Ex_{\bx\sim\mathcal{D}}[\chi_S(\bx)f(\bx)].\tag{$\hat{T}(S)\le 1$}
\begin{lemma}[Obtaining a zero-error parity for $f$ from a well-correlated parity for $(f^{\mathrm{ext}})_{\oplus \ell}$]
\label{lem:zero-error}
    Let $\mathcal{D}$ and $(f^{\mathrm{ext}})_{\oplus \ell}:(\F_2^\ell)^n\to\F_2$ be as in the statement of \Cref{claim:search}. If there is a parity $\chi_S$ for $S\sse [\ell n]$ such that $\dist_{\mathcal{D}}(\chi_S,(f^{\mathrm{ext}})_{\oplus \ell})\le \frac{1}{2}-\gamma$ for $\gamma>0$, then $\chi_{S^\star}(y)=f(y)$ for all $y\in D$ where $|S^\star|\le |S|/\ell$ and $S^\star$ consists of the coordinates $i\in [n]$ such that the $i$th block in $S$ is nonempty.
\end{lemma}

\begin{proof}
    %The contrapositive of \Cref{lem:parity-sub-amplification} implies that there is a parity $S'$ of size $|S|/\ell$ satisfying $\dist_{\mathrm{Unif}(\mathrm{Span}({D}))}(\chi_{S'},(f^{\mathrm{ext}})_{\oplus \ell})\le \frac{1}{2}-\gamma$. Inspecting the proof of \Cref{lem:parity-sub-amplification}, we see that $S'$ consists of the coordinates $i\in [n]$ such that the $i$th block in $S$ is nonempty. Finally, the proof of the no case in \Cref{lem:boosting} shows that $\chi_{S'}(y)=f(y)$ for all $y\in D$. Indeed, if it were the case that $\chi_{S'}$ disagrees with $f$ on some input $y\in D$, then the proof of \Cref{lem:boosting} shows that $\dist_{\mathrm{Unif}(\mathrm{Span}({D}))}(\chi_{S'},(f^{\mathrm{ext}})_{\oplus \ell})= \tfrac{1}{2}$ which contradicts our assumption on the error of $\chi_{S'}$. 
    \Cref{lem:new parity amplification} states that there is a parity $S^\star$ of size $|S|/\ell$ satisfying $\dist_{\mathrm{Unif}(\mathrm{Span}({D}))}(\chi_{S^\star},f^{\mathrm{ext}})\le \frac{1}{2}-\gamma$. Further, $S^\star$ consists of the coordinates $i\in [n]$ such that the $i$th block in $S$ is nonempty. Finally, the contrapositive of the no case in \Cref{lem:boosting} implies that $\chi_{S^\star}(y)=f(y)$ for all $y\in D$. Indeed, if it were the case that $\chi_{S^\star}$ disagrees with $f$ on some input $y\in D$, then \Cref{lem:boosting} shows that $\dist_{\mathrm{Unif}(\mathrm{Span}({D}))}(\chi_{S^\star},f^{\mathrm{ext}})= \tfrac{1}{2}$ which contradicts our assumption on the error of $\chi_{S^\star}$.\footnote{We are using the fact that implicit in the proof of \Cref{lem:boosting} is the following: for any parity $\chi_S$, if $\chi_S$ disagrees with $f$ on at least one $x\in D$, then $\chi_S$ disagrees with $f^{\mathrm{ext}}$ on exactly half of the inputs from $\mathrm{Span}(D)$.  
    %Technically, \Cref{lem:boosting} states that if $f$ disagrees with \emph{every} size $|S^\star|$ parity on some input, then $\dist_{\mathrm{Unif}(\mathrm{Span}({D}))}(\chi_{S^\star},f^{\mathrm{ext}})= \tfrac{1}{2}$. However, the same conclusion also holds with the weaker assumption that $f$ disagrees with $\chi_{S^\star}$ on some input. In fact, this is what we show in the proof of \Cref{lem:boosting}.
    }
\end{proof}

\subsection{Proof of \Cref{claim:search}}

First, we prune all paths in $T$ at depth $c\log s$ to obtain a tree $T'$. \Cref{claim:pruning-depth} ensures that this process doesn't increase the error of $T'$ too much:
\begin{align*}
    \dist_{\mathcal{D}}(T',(f^{\mathrm{ext}})_{\oplus \ell})&\le \dist_{\mathcal{D}}(T,(f^{\mathrm{ext}})_{\oplus \ell}))+s^{1-c(1-1/\ell)}\tag{\Cref{claim:pruning-depth}}\\
    &\le \frac{1}{2}-\gamma+s^{1-c(1-1/\ell)}.\tag{Assumption on $T$}
\end{align*}
% \begin{align*}
%     \frac{1}{2}-\gamma+s^{1-c(1-1/\ell)}&\ge \dist_{\mathcal{D}}(T,(f^{\mathrm{ext}})_{\oplus \ell}))+s^{1-c(1-1/\ell)}\tag{Assumption on $T$}\\
%     &\ge \dist_{\mathcal{D}}(T',(f^{\mathrm{ext}})_{\oplus \ell})\tag{\Cref{claim:pruning-depth}}
% \end{align*}
After pruning, $T'$ has depth small enough that, in polynomial time, we can apply \Cref{lem:extraction} to obtain a well-correlated parity $\chi_S$ of size $\le c\log s$. The error of this parity is bounded:
\begin{equation}
\label{eq:par-error}
\dist_{\mathcal{D}}(\chi_S,(f^{\mathrm{ext}})_{\oplus \ell})\leq \frac{1}{2}-\Theta\bigg(\frac{\gamma-s^{1-c(1-1/\ell)}}{s^{2c}}\bigg).
\end{equation}
By our assumption that $\gamma\ge\Omega(s^{1-c(1-1/\ell)})$, \Cref{eq:par-error} can be rewritten as $\dist_{\mathcal{D}}(\chi_S,(f^{\mathrm{ext}})_{\oplus \ell})\le \frac{1}{2}-\gamma'$ for some $\gamma'>0$. Therefore, \Cref{lem:zero-error} implies that we can find a parity $S^\star$ of size $\le \frac{c\log s}{\ell}$ such that $\chi_{S^\star}(y)=f(y)$ for all $y\in D$ as desired.\hfill\qed

\subsection{Proof of \Cref{thm:most general}}
By \Cref{thm:general main reduction}, for any $\alpha>1$, given an NCP instance where the nearest codeword is within distance $k$ of the received word, there is an algorithm running in time $O(\ell n)\cdot t(\ell n,2^{\ell k},2^{O(\alpha \ell k)},\eps)$ for $\eps=\frac{1}{2}-2^{-\Omega(\alpha \ell k)}$ which outputs a decision tree of size $2^{O(\alpha \ell k)}$ for $(f^{\mathrm{ext}})_{\oplus \ell}$ and has error $\eps$ in computing $(f^{\mathrm{ext}})_{\oplus \ell}$ over $\mathcal{D}=\mathrm{Unif}(\mathrm{Span}({D}))_{\oplus \ell}$. Therefore, by \Cref{claim:search} we can extract a parity of size $|S|\le \alpha k$ which is consistent with $f$ over $D$. Equivalently, we have found a codeword within distance $\alpha k$ of the received word as desired. Since this extraction step requires an additional $\poly(n,\ell,2^{\alpha \ell k})$ time, the proof is completed.\hfill\qed

% After applying \Cref{lem:extraction}, we obtain with high probability a parity $\chi_S$ such that $\Pr_{\bx\sim\mathcal{D}}[\chi_S(\bx)\neq f(\bx)]\le \frac{1}{2}-\gamma'$

% We then apply \Cref{lem:extraction} to $T'$ to obtain a well-correlated parity $\chi_S$ of size $\le c\log s$. Finally, we apply \Cref{lem:zero-error} to obtain a parity $S'$ of size $\le c\log s/\ell$. It remains to analyze the error of this parity:

% \red{TODO:
% Rest of the proof proceeds as follows
% \begin{enumerate}
%     \item Take the decision tree $T$ for $(f^{\mathrm{ext}})_{\oplus \ell}$ and prune the depth to $\log s$
%     \item Using \Cref{lem:extraction}, find a well correlated parity of size $\log s$
%     \item Using what we proved about parity substitution, collapse the blocks to find a well-correlated parity with $f^{\mathrm{ext}}$ of size $\log s/\ell$
%     \item Using what we proved about the Span operation, show that the well-correlated parity must have $0$-error over $D$.
% \end{enumerate}
% }
\section{Proofs of corollaries}

\subsection{Proof of \Cref{cor:improving EH gives new NCP algorithms}}
Let $\mathcal{A}$ be the learner from the corollary statement. Using \Cref{thm:most general} with $\ell=2$, we show that $\mathcal{A}$ solves $O(\log n)$-approximate $k$-NCP. Given a decision tree target of size $2^{2k}$ and random labeled examples from $\mathcal{D}$, $\mathcal{A}$ runs in time $n^{o(k)}$ and outputs a decision tree hypothesis with accuracy $\frac{1}{2}+\frac{1}{\poly(n)}$. If $\alpha=O(\log n)$, then the size of the decision tree hypothesis is at most $n^{o(k)}\le 2^{O(\alpha k)}$ and the error of the hypothesis satisfies $\eps=\frac{1}{2}-\frac{1}{\poly(n)}\le \frac{1}{2}-2^{-\Omega(\alpha k)}$. Therefore, \Cref{thm:most general} shows that $\mathcal{A}$ solves $O(\log n)$-approximate $k$-NCP for $k=\Theta(\log s)$.

\subsection{Proof of \Cref{cor:W1 hardness of polytime}}
Suppose for contradiction there is a learner $\mathcal{A}$ satisfying the constraints of the corollary statement. We will use $\mathcal{A}$ to solve $c$-approximate $k$-NCP for some constant $c>1$ in randomized polynomial-time. By \Cref{thm:w1-hardness}, this implies that there is a randomized FPT algorithm for all of $\mathrm{W}[1]$. 

Let $c'$ be a constant so that $\mathcal{A}$ runs in time $n^{c'}$ when given random examples from $\mathcal{D}$ labeled by a size-$n$ decision tree and outputs a hypothesis with error $\frac{1}{2}-\frac{1}{n^{c'}}$ under $\mathcal{D}$. Let $c$ be a sufficiently large constant relative to $c'$ (to be chosen later). We will use $\mathcal{A}$ to solve $c$-approximate $k$-NCP over $\F_2^n$. We assume that $n$ is large enough so that $\log n\ge k$. Let $\ell=(\log n)/k$. Given a decision tree target of size $2^{\ell k}=n$, $\mathcal{A}$ runs in time $n^{c'}$ and outputs a decision tree hypothesis of size at most $n^{c'}\le 2^{O(c \ell k)}=n^{O(c)}$, assuming $c$ is a large enough. Likewise, the error of the hypothesis is at most $\frac{1}{2}-n^{-c'}\le \frac{1}{2}-2^{-\Omega(c\ell k)}=\frac{1}{2}-n^{-\Omega(c)}$, again assuming that $c$ is large enough. By \Cref{thm:most general}, this shows that $\mathcal{A}$ solves $c$-approximate $k$-NCP in $\poly(n)$ time as desired.

\subsection{Proof of \Cref{cor:explicit superpoly}}
Let $\mathcal{A}$ be the learner from the corollary statement. Let $c'>1$ be a constant such that $\mathcal{A}$ learns decision tree targets of size $n$ with decision tree hypotheses of size $n^{c'}$. We start by proving ETH hardness.

\paragraph{ETH hardness.}
Combining \cite[Theorem 1]{LLL24} and \cite[Theorem 11]{LRSW22} yields the following reduction from 3-SAT to $k$-NCP.
\begin{theorem}[Reduction from solving 3-SAT exactly to approximating $k$-NCP]
\label{thm:lll reduction}
    For all constant $c>1$, there is a constant $q>1$ such that for all $k\in\N$ the following holds. There is a reduction running in time $\poly(m,2^{m/k})+\poly(m,2^{k})$ which maps $3$-\textnormal{SAT} instances $\varphi$ consisting of $m$ clauses to \textnormal{NCP} instances $(G,z)$ of size $\poly(m,2^{m/k})$ such that
    \begin{itemize}
        \item[$\circ$] \textnormal{Yes} case: if $\varphi$ is satisfiable then $(G,z)$ is a \textnormal{Yes} instance of $c$-approximate $k^q$-\textnormal{NCP};
        \item[$\circ$] \textnormal{No} case: if $\varphi$ is not satisfiable, then $(G,z)$ is a \textnormal{No} instance of $c$-approximate $k^q$-\textnormal{NCP}.
    \end{itemize}
\end{theorem}

Using \Cref{thm:lll reduction}, we show how to refute randomized ETH if $\mathcal{A}$ runs in time $n^{(\log n)^{\delta}}$ for sufficiently small $\delta>0$.  Let $\varphi$ be a 3-SAT instance on $n$ variables with $m$ clauses. By \Cref{thm:lll reduction}, for a constant $c>1$ (which is sufficiently larger than $c'$ and is chosen later), there is a constant $q>1$ such that the reduction holds for all $k\in\N$. Let $k=m^{\lambda}$ for any $0<\lambda<1/(2q)$ and let $(H,t)\in\F_2^{M\times N}\times \F_2^M$ for $M+N=\poly(m,2^{m/k})= 2^{O(m^{1-\lambda})}$ be the $k^q$-NCP instance from \Cref{thm:lll reduction}. To refute randomized ETH, it is sufficient to solve $c$-approximate $k^q$-NCP with respect to $(H,t)$ in randomized time $2^{o(m)}$.
Assume that $\delta$ is small enough so that $(1-\lambda)(1+\delta)<1$. We claim that by \Cref{thm:most general} with $\ell=2$, the learner $\mathcal{A}$ solves the $k^q$-NCP instance in the desired amount of time. 

Given a decision tree target of size $2^{2k^q}$ over $2N$ variables, $\mathcal{A}$ runs in time $(2N)^{(\log 2N)^\delta}= 2^{O(m^{(1-\lambda)(1+\delta)})}=2^{o(m)}$ by our assumption on $\delta$. We use here the fact that the size of the target satisfies $2^{2k^q}=2^{2m^{\lambda q}}\le 2N$ by our choice of $\lambda$. Moreover, $\mathcal{A}$ outputs a decision tree hypothesis of size $(2N)^{c'}\le 2^{O(ck^q)}$ with error $\frac{1}{2}-\frac{1}{N^{c'}}\le \frac{1}{2}-2^{-\Omega(ck^q)}$ for sufficiently large $c$. By \Cref{thm:most general}, this shows that $\mathcal{A}$ solves $c$-approximate $k$-NCP with high probability in $2^{o(m)}$ time as desired.

\paragraph{Gap-ETH hardness.}
The following reduction is implicit in \cite{Man20} by stringing together the reduction from 3-SAT to {\sc Label Cover} (\cite[Theorem 9]{Man20}), and from {\sc Label Cover} to NCP (\cite[Corollary 5]{Man20}).

\begin{theorem}[Reduction from gapped 3-SAT to approximating $k$-NCP]
\label{thm:manurangsi reduction}
    For all constants $c>1$ and $\lambda>0$, and for every $k\in\N$, there is a reduction running in time $\poly(k,m,2^{m/k})$ which maps 3-\textnormal{SAT} instances $\varphi$ consisting of $m$ clauses to \textnormal{NCP} instances $(G,z)$ of size $\poly(k,m,2^{m/k})$ such that 
    \begin{itemize}
        \item[$\circ$] \textnormal{Yes} case: if $\varphi$ is satisfiable then $(G,z)$ is a \textnormal{Yes} instance of $c$-approximate $k$-\textnormal{NCP};
        \item[$\circ$] \textnormal{No} case: if every assignment to $\varphi$ satisfies at most $(1-\lambda)m$ clauses, then $(G,z)$ is a \textnormal{No} instance of $c$-approximate $k$-\textnormal{NCP}.
    \end{itemize}
\end{theorem}

Using \Cref{thm:manurangsi reduction}, we show how to refute randomized Gap-ETH if $\mathcal{A}$ runs in time $n^{o(\log n)}$. Let $\varphi$ be a 3-SAT instance on $n$ variables and with $m$ clauses and let $\lambda>0$ be given. Using \Cref{thm:manurangsi reduction} with $k=\sqrt{m}$ and for $c$ larger than $c'$ (to be specified later), we obtain a $c$-approximate $k$-NCP instance $(H,t)\in\F_2^{M\times N}\times\F_2^M$ where $H$ is the parity check matrix for a linear code and $M+N=\poly(k,m,2^{O(m/k)})=2^{O(\sqrt{m})}$. Note in particular we can assume $2^{2k}\le 2N$ (this will satisfy our assumption on the size of the target decision tree). By \Cref{thm:manurangsi reduction}, to approximate the number of satisfiable clauses of $\varphi$, it is sufficient to solve $c$-approximate $k$-NCP on $(H,t)$ in randomized time $2^{o(m)}$. We claim that by \Cref{thm:most general} with $\ell=2$, the learner $\mathcal{A}$ solves the $k$-NCP instance in the desired amount of time. 

Given a decision tree target of size $2^{2k}$ over $2N$ variables, $\mathcal{A}$ runs in time $(2N)^{o(\log 2N)}=(2^{O(\sqrt{m})})^{o(\sqrt{m})}=2^{o(m)}$. Moreover, $\mathcal{A}$ outputs a decision tree hypothesis of size $(2N)^{c'}\le 2^{O(ck)}$ with error $\frac{1}{2}-\frac{1}{N^{c'}}\le \frac{1}{2}-2^{-\Omega(ck)}$ for sufficiently large $c$. By \Cref{thm:most general}, this shows that $\mathcal{A}$ solves $c$-approximate $k$-NCP with high probability in $2^{o(m)}$ time as desired.\hfill\qed

\section*{Acknowledgments}

% 

% \acks{

We thank the FOCS reviewers for their helpful feedback and suggestions.

The authors are supported by NSF awards 1942123, 2211237, 2224246, a Sloan Research Fellowship, and a Google Research Scholar Award. Caleb is also supported by an NDSEG Fellowship, and Carmen by a Stanford Computer Science Distinguished Fellowship and an NSF Graduate Research Fellowship.

% \cleardoublepage
% \section{\red{Unincorporated material}}
% \input{unincorporated}

% \section*{Acknowledgments}

% 

% Caleb, Carmen, and Li-Yang are supported by NSF awards 1942123, 2211237, and 2224246. Caleb is also supported by an NDSEG fellowship. 

 \bibliography{ref}
 \bibliographystyle{alpha}

\end{document}